\documentclass[11pt]{amsart}
\usepackage{amsbsy,amssymb,amsmath,amsthm,amscd,amsfonts,latexsym,amstext,delarray,
amsmath,graphicx}
\usepackage[margin=1in]{geometry}
\usepackage{color}
\input xypic

\newtheorem{thm}{Theorem}[section]
\newtheorem{prop}[thm]{Proposition}
\newtheorem{cor}[thm]{Corollary}
\newtheorem{lem}[thm]{Lemma}

\newtheorem{defn}[thm]{Definition}
\newtheorem{rem}[thm]{Remark}
\newtheorem{ex}[thm]{Example}
\newtheorem{ques}[thm]{Question}

\numberwithin{equation}{section}

\def\F{{\mathbb F}}
\def\Q{{\mathbb Q}}
\def\Z{{\mathbb Z}}
\def\N{{\mathbb N}}
\def\R{{\mathbb R}}
\def\C{{\mathbb C}}
\renewcommand\P{{\mathbb P}}
\def\H{{\mathbb H}}

\def\bG{{\mathbb G}}
\def\P{{\mathbb P}}
\def\A{{\mathbb A}}
\def\GL{{\rm GL}}
\def\bL{{\mathbb L}}

\def\cC{{\mathcal C}}
\def\cD{{\mathcal D}}

\def\cH{{\mathcal H}}

\def\cL{{\mathcal L}}
\def\cM{{\mathcal M}}

\def\cO{{\mathcal O}}
\def\cP{{\mathcal P}}

\def\cR{{\mathcal R}}
\def\cS{{\mathcal S}}
\def\cT{{\mathcal T}}
\def\cU{{\mathcal U}}
\def\cV{{\mathcal V}}

\def\cX{{\mathcal X}}
\def\cY{{\mathcal Y}}
\def\cZ{{\mathcal Z}}

\def\bF{{\mathbb F}}
\def\bG{{\mathbb G}}

\def\bL{{\mathbb L}}

\DeclareMathOperator*{\Hom}{Hom}

\def\m{{\mathfrak m}}

\def\GL{{\rm GL}}
\def\PGL{{\rm PGL}}

\title[Rota--Baxter algebras, singular hypersurfaces, renormalization]{Rota--Baxter algebras,
singular hypersurfaces, and renormalization on Kausz compactifications}
\author{Matilde Marcolli and Xiang Ni}
\address{Department of Mathematics, Division of Mathematics, Physics and Astronomy,
California Institute of Technology, 1200 E. California Blvd. Pasadena, CA 91125, USA}
\email{matilde@caltech.edu}
\email{xni@caltech.edu}
\date{}

\begin{document}
\maketitle

\begin{abstract}
We consider Rota-Baxter algebras of meromorphic forms with poles along
a (singular) hypersurface in a smooth projective variety and the associated
Birkhoff factorization for algebra homomorphisms from a commutative Hopf
algebra. In the case of a normal crossings divisor,
the Rota-Baxter structure simplifies considerably and the factorization becomes
a simple pole subtraction.
We apply this formalism to the unrenormalized momentum space
Feynman amplitudes, viewed as (divergent) integrals in the complement of
the determinant hypersurface. We lift the integral to the Kausz compactification
of the general linear group, whose boundary divisor is normal crossings.
We show that the Kausz compactification is a Tate motive and that the boundary
divisor and the divisor that contains the boundary of the chain of integration are 
mixed Tate configurations. The regularization of the integrals that we obtain
differs from the usual renormalization of physical Feynman amplitudes,
and in particular it may give mixed Tate periods in some cases that have non-mixed Tate
contributions when computed with other renormalization methods.
\end{abstract}

\section{Introduction}

In this paper, we consider the problem of extracting periods of algebraic varieties
from a class of divergent integrals arising in quantum field theory. The method we
present here provides a regularization and extraction of finite values that differs
from the usual (renormalized) physical Feynman amplitudes, but whose mathematical
interest lies in the fact that it gives a period of a mixed Tate motive, for
all graphs for which the amplitude
can be computed using (global) forms with logarithmic poles. For more general
graphs, one also obtains a period, where the nature of the motive involved
depends on how a certain hyperplane arrangement intersects the big
cell in a compactification of the general linear group. More precisely, the motive considered here 
is provided by the Kausz compactification of the general linear group and by a
hyperplane arrangement that contains the boundary of the chain of integration. 
The regularization procedure we propose is modeled on the algebraic renormalization method,
based on Hopf algebras of graphs and Rota--Baxter algebras, as originally developed
by Connes and Kreimer \cite{CoKr} and by Ebrahmi-Fard, Guo, and Kreimer \cite{EFGK}.
The main difference in our approach is that we apply the formalism to a Rota--Baxter algebra of (even)
meromorphic differential forms instead of applying it to a regularization of the integral.
The procedure becomes especially simple in cases where the de Rham cohomology
of the singular hypersurface complement is all realized by forms with logarithmic poles,
in which case we replace the divergent integral with a family of convergent integrals
obtained by a pole subtraction on the form and by (iterated) Poincar\'e residues.
A similar approach was developed for integrals in configuration spaces
by Ceyhan and the first author \cite{CeMa}.

\smallskip

In Section \ref{RBmeromSec} we introduce Rota--Baxter algebras of even
meromorphic forms, along the lines of \cite{CeMa},
and we formulate a general setting for extraction of finite values (regularization
and renormalization) of divergent integrals modeled on algebraic renormalization
applied to these Rota--Baxter algebras of differential forms.

In Section \ref{RBlogSec} we discuss the Rota--Baxter algebras of even
meromorphic forms in the case of a smooth hypersurface $Y\subset X$.
We show that, when restricted to forms with logarithmic poles, the
Rota--Baxter operator becomes simply a derivation, and the Birkhoff
factorization collapses to a simple pole subtraction, as in the case of
log divergent graphs. We show that this simple pole subtraction can
lead to too much loss of information about the unrenormalized
integrand and we propose considering the additional information of
the Poincar\'e residue and an additional integral associated to the
residue.

In Section \ref{SingHypSec} we consider the case of singular
hypersurfaces $Y\subset X$ given by a simple normal crossings
divisor. We show that, in this case, the Rota--Baxter operator
satisfies a simplified form of the Rota--Baxter identity, which
however is not just a derivation. We show that this modified
identity still suffices to have a simple pole subtraction $\phi_+(\Gamma)=(1-T)\phi(\Gamma)$
in the Birkhoff factorization, even though the negative piece
$\phi_-(\Gamma)$ becomes more complicated. Again, to avoid too much loss of information
in passing from $\phi(\Gamma)$ to $\phi_+(\Gamma)$, we consider, in addition to
the renormalized integral $\int_\sigma \phi_+(\Gamma)$, the collection
of integrals of the form $\int_{\sigma\cap Y_I} {\rm Res}_{Y_I}(\phi(\Gamma))$,
where ${\rm Res}_{Y_I}$ is the iterated Poincar\'e residue, \cite{Del},
along the intersection $Y_I=\cap_{j\in I}Y_j$
of components of $Y$. These integrals are all periods of mixed
Tate motives if $\{ Y_I \}$ is a mixed Tate configuration, in the sense of \cite{Gon}.
We discuss the question of further generalizations to more
general types of singularities, beyond the normal crossings case,
via Saito's theory of forms with logarithmic poles \cite{Sa}, by showing
that one can also define a Rota--Baxter structure on the Saito
complex of forms with logarithmic poles.

In Section \ref{KGLsec} we present our main application, which
is a regularization (different from the physical one) of the Feynman
amplitudes in momentum space, computed on the complement of
the determinant hypersurface as in \cite{AluMa}. Since the
determinant hypersurface has worse singularities than what we
need, we pull back the integral computation to the Kausz
compactification \cite{Kausz} of the general linear group, where
the boundary divisor that replaces the determinant hypersurface is
a simple normal crossings divisor. We show that the 
motive of the Kausz compactification is Tate, and that the
components of the boundary divisor form a mixed Tate configuration.
We discuss how one can replace the form $\eta_\Gamma$
of the Feynman amplitude with a form with logarithmic poles. In general, 
this form is defined on the big cell of the Kausz compactification.
For certain graphs, it is possible to show, using the mixed Hodge
structure, that the form with logarithmic poles extends globally
to the Kausz compactification, with poles along the boundary divisor.

\section{Rota--Baxter algebras of meromorphic forms}\label{RBmeromSec}

We generalize the algebraic renormalization formalism to a setting
based on Rota--Baxter algebras of algebraic differential forms on a
smooth projective variety with poles along a hypersurface.

\subsection{Rota--Baxter algebras}

A Rota--Baxter algebra of weight $\lambda$ is a unital commutative algebra
$\cR$ over a field $K$ such that $\lambda\in K$, together with a linear operator 
$T:\cR \to \cR$ satisfying the Rota--Baxter identity
\begin{equation}\label{RBlambda}
 T(x) T(y) = T(x T(y))+ T(T(x)y) +\lambda T(xy) .
\end{equation}

\smallskip

For example, Laurent polynomials $\cR=\C[z,z^{-1}]$ with $T$ the projection
onto the polar part are a Rota--Baxter algebra of weight $-1$.

\smallskip

The Rota--Baxter operator $T$ of a Rota--Baxter algebra of weight $-1$,
satisfying
\begin{equation}\label{RBmin1}
T(x) T(y) + T(xy) = T(x T(y))+ T(T(x)y) ,
\end{equation}
determines a splitting of $\cR$ into $\cR_+=(1-T)\cR$ and $T\cR$,
where $(1-T)\cR$ and $T\cR$ are not just vector spaces but algebras,
because of the Rota--Baxter relation \eqref{RBmin1}.
The algebra $T\cR$ is non-unital. In order to work with unital algebras,
one defines $\cR_-$ to be the unitization of $T\cR$, that is, $T\cR\oplus K$
with multiplication $(x,t)(y,s)=(xy+ty+sx, ts)$. 
For an introduction to Rota--Baxter algebras we refer the reader to \cite{Guo}.

\smallskip
\subsection{Rota--Baxter algebras of even meromorphic forms}

Let $Y$ be a hypersurface in a
projective variety $X$, with defining equation $Y=\{ f=0 \}$.
We denote by $\cM^\star_{X}$ the sheaf of meromorphic differential forms
on $X$, and by $\cM^\star_{X,Y}$ the subsheaf  of meromorphic forms
on with poles (of arbitrary order) along $Y$, that is, $\cM^\star_{X,Y}=j_*\Omega^1_U$,
where $U=X\smallsetminus Y$ and $j: U \hookrightarrow X$ is the inclusion.  
Passing to global sections of $\cM^\star_{X,Y}$ gives a graded-commutative
algebra over the field of definition of the varieties $X$ and $Y$, which, for simplicity, 
we will still denote by $\cM^\star_{X,Y}$.
We can write forms $\omega\in \cM^\star_{X,Y}$
as sums $\omega =\sum_{p\geq 0} \alpha_p / f^p$, where the $\alpha_p$
are holomorphic forms.

\smallskip

In particular, we consider forms of even degrees, so that $\cM^{\rm even}_{X,Y}$
is a commutative algebra under the wedge product. 

\begin{lem}\label{RBmerom}
The commutative algebra $\cM^{\rm even}_{X,Y}$, together with the linear operator
$T: \cM^{\rm even}_{X,Y} \to \cM^{\rm even}_{X,Y}$ defined as the polar part
\begin{equation}\label{Tmeromf}
T(\omega)= \sum_{p\geq 1} \alpha_p / f^p,
\end{equation}
is a Rota--Baxter algebra of weight $-1$.
\end{lem}

\proof For $\omega_1=\sum_{p\geq 0} \alpha_p / f^p$ and $\omega_2=\sum_{q\geq 0} \beta_q/f^q$, we have
$$ T(\omega_1\wedge \omega_2)=\sum_{p\geq 0, q\geq 1} \frac{\alpha_p \wedge \beta_q}{f^{p+q}}+\sum_{p\geq 1, q\geq 0}
\frac{\alpha_p \wedge \beta_q}{f^{p+q}}-\sum_{p\geq 1, q\geq 1} \frac{\alpha_p \wedge \beta_q}{f^{p+q}},$$
$$ T(T(\omega_1)\wedge \omega_2)=\sum_{p\geq 1, q\geq 0} \frac{\alpha_p \wedge \beta_q}{f^{p+q}}, $$
$$ T( \omega_1 \wedge T(\omega_2))= \sum_{p\geq 0, q\geq 1} \frac{\alpha_p \wedge \beta_q}{f^{p+q}}, $$
$$ T(\omega_1)\wedge T(\omega_2)=\sum_{p\geq 1, q\geq 1} \frac{\alpha_p \wedge \beta_q}{f^{p+q}}, $$
so that \eqref{RBmin1} is satisfied. 
\endproof

Note that the restriction to even form is introduced only in order to ensure that the resulting Rota--Baxter algebra is commutative, while \eqref{Tmeromf} satisfies \eqref{RBmin1} regardless of the restriction on degrees.

\begin{rem}\label{RBomegaY} {\rm 
Equivalently, we have the following description of the Rota--Baxter operator, which we
will use in the following. The linear operator
\begin{equation}\label{TomegaY}
T(\omega)= \alpha \wedge \xi,  \ \ \   \text{ for } \omega =\alpha \wedge \xi + \eta,
\end{equation}
acting on forms $ \omega =\alpha \wedge \xi + \eta$, with $\alpha$ a meromorphic
form on $X$ with poles on $Y$ and $\xi$ and $\eta$ holomorphic forms on $X$, is a
Rota--Baxter operator of weight $-1$. }
\end{rem}

The Rota--Baxter identity is equivalently seen then as follows.
For $\omega_i=\alpha_i\wedge \xi_i + \eta_i$, with $i=1,2$, we have
$$ T(\omega_1\wedge \omega_2)=(-1)^{|\alpha_2|\,|\xi_1|}
\alpha_1 \wedge \alpha_2 \wedge \xi_1 \wedge \xi_2 + \alpha_1 \wedge \xi_1 \wedge \eta_2
+(-1)^{|\eta_1|\, |\alpha_2|} \alpha_2 \wedge \eta_1 \wedge \xi_2 $$
while
$$ T(T(\omega_1)\wedge \omega_2)= (-1)^{|\alpha_2|\,|\xi_1|}
\alpha_1 \wedge \alpha_2 \wedge \xi_1 \wedge \xi_2 + \alpha_1 \wedge \xi_1 \wedge \eta_2 $$
$$ T(\omega_1 \wedge T(\omega_2))= (-1)^{|\alpha_2|\,|\xi_1|}
\alpha_1 \wedge \alpha_2 \wedge \xi_1 \wedge \xi_2
+(-1)^{|\eta_1|\, |\alpha_2|} \alpha_2 \wedge \eta_1 \wedge \xi_2 $$
and
$$ T(\omega_1) \wedge T(\omega_2)=(-1)^{|\alpha_2|\,|\xi_1|}
\alpha_1 \wedge \alpha_2 \wedge \xi_1 \wedge \xi_2, $$
where all signs are positive if the forms are of even degree.
Thus, the operator $T$ satisfies \eqref{RBmin1}.

\smallskip

The proof automatically extends to the following slightly more general setting.

\begin{lem}\label{prodRB}
Let $(X_\ell, Y_\ell)$ for $\ell\geq 1$ be a collection of smooth projective
varieties $X_\ell$ with hypersurfaces $Y_\ell$, all defined over the same field. 
Then the commutative algebra $\bigwedge_\ell \cM^{\rm even}_{X_\ell, Y_\ell}$
is a Rota--Baxter algebra of weight $-1$ with the polar projection operator $T$
determined by the $T_\ell$ on each $\cM^{\rm even}_{X_\ell, Y_\ell}$.
\end{lem}

A similar setting was considered in Theorem 6.4 of \cite{CeMa}.

\subsection{Renormalization via Rota--Baxter algebras}\label{renRBsec}

In \cite{CoKr}, the BPHZ renormalization procedure of perturbative
quantum field theory was reinterpreted as a Birkhoff factorization of
loops in the pro-unipotent group of characters of a commutative
Hopf algebra of Feynman graphs. This procedure of {\em algebraic
renormalization} was reformulated in more general and abstract
terms in \cite{EFGK}, using Hopf algebras and Rota--Baxter algebras.

\smallskip

We summarize here quickly the basic setup of algebraic renormalization.
We refer the reader to \cite{CoKr}, \cite{CoMa}, \cite{EFGK}, \cite{Mar} for
more details.

\smallskip

The Connes--Kreimer Hopf algebra of Feynman graphs $\cH$ is
a commutative, non-cocommutative, graded, connected Hopf 
algebra over $\Q$ associated to a given Quantum Field Theory (QFT).
A theory is specified by assigning a Lagrangian and the corresponding action functional, 
which in turn determines which graphs occur as Feynman graphs of the theory. For
instance, the only allowed valences of vertices in a Feynman graph are the powers
of the monomials in the fields that appear in the Lagrangian. The generators
of the Connes--Kreimer Hopf algebra of a given QFT are the 1PI Feynman
graphs $\Gamma$ of the theory, namely those Feynman graphs that are 2-egde
connected. As a commutative algebra,  
$\cH$ is then just a polynomial algebra in the 1PI graphs $\Gamma$. A grading
on $\cH$ is given by the loop number (first Betti number) of graphs.
In the case where Feynman graphs also
have vertices of valence $2$, one uses the number of internal edges instead of
loop number, to have finite dimensional graded pieces, but we ignore this subtlety
for the present purposes. The grading satisfies 
$$ \deg(\Gamma_1\cdots\Gamma_n)=\sum_i \deg(\Gamma_i),
\ \ \ \deg(1)=0. $$
The connectedness property means that the degree zero part is just $\Q$. 
The coproduct in $\cH$ is given by 
\begin{equation}\label{CKcoprod}
 \Delta(\Gamma)=\Gamma\otimes 1 + 1 \otimes \Gamma +
\sum_{\gamma \in \cV(\Gamma)}  \gamma \otimes \Gamma/\gamma,
\end{equation}
where the class $\cV(\Gamma)$ consists of all (not necessarily connected)
divergent subgraphs $\gamma$ such that the quotient graph (identifying each
component of $\gamma$ to a vertex) is still a 1PI  Feynman graph of the
theory. As in any graded connected Hopf algebra, the antipode is constructed inductively as
$$ S(\Gamma)=-\Gamma -\sum S(\Gamma') \Gamma'' $$
for $\Delta(\Gamma)=\Gamma \otimes 1 + 1 \otimes \Gamma + \sum \Gamma' \otimes \Gamma''$,
with the terms $\Gamma'$, $\Gamma''$ of lower degrees.

\smallskip

\begin{rem}\label{notationGamma} {\rm 
The general element in the Hopf algebra $\cH$ is not a graph $\Gamma$ but a 
polynomial function $P=\sum a_{i_1,\ldots,i_k} \Gamma_{i_1}^{n_{i_1}}\cdots \Gamma_{i_k}^{n_{i_k}}$ 
with $\Q$ coefficients in the generators given by the graphs. However, for simplicity of notation, in the 
following we will just write $\Gamma$ to denote an arbitrary element of $\cH$. }
\end{rem}

\smallskip

An {\em algebraic Feynman rule} $\phi: \cH \to \cR$ is a {\em homomorphism of
commutative algebras} from the Hopf algebra $\cH$ of Feynman graphs
to a Rota--Baxter algebra $\cR$ of weight $-1$,
$$ \phi \in {\rm Hom}_{\rm Alg}(\cH,\cR). $$
The set ${\rm Hom}_{\rm Alg}(\cH,\cR)$ has a group structure, where the multiplication $\star$ 
is dual to the coproduct in the Hopf algebra, 
$\phi_1\star \phi_2(\Gamma) = \langle \phi_1\otimes \phi_2, \Delta(\Gamma)\rangle$.

\smallskip

Algebra homomorphisms $\phi:\cH\rightarrow \cR$ between a Hopf algebra 
$\cH$ and a Rota--Baxter algebra $\cR$ are also often referred to as ``characters" in the
renormalization literature.

\smallskip

The morphism $\phi$ by itself does not know about the coalgebra structure
of $\cH$ and the Rota--Baxter structure of $\cR$. These enter in the factorization
of $\phi$ into divergent and finite part.

\smallskip

A Birkhoff factorization of an algebraic Feynman rule consists of
a pair of commutative algebra homomorphisms
$$ \phi_\pm \in {\rm Hom}_{\rm Alg}(\cH,\cR_\pm) $$
where $\cR_\pm$ is the splitting of $\cR$ induced by the Rota--Baxter
operator $T$, with $\cR_+=(1-T)\cR$ and $\cR_-$ the unitization of $T\cR$,
satisfying
$$ \phi = (\phi_-\circ S)\star \phi_+, $$
with the product $\star$ dual to the coproduct $\Delta$ as above. 
The Birkhoff factorization is unique if one also imposes the normalization
condition $\epsilon_-\circ \phi_-=\epsilon$, where $\epsilon$ is the counit
of $\cH$ and $\epsilon_-$ is the augmentation in the algebra $\cR_-$.

\smallskip

As shown in Theorem 4 of \cite{CoKr} (see equations (32) and (33) therein), there is 
an inductive formula for the Birkhoff factorization of an algebraic Feynman rule, of the form
\begin{equation}\label{recBirkhoff}
 \phi_-(\Gamma)=-T (\phi(\Gamma) +\sum \phi_-(\Gamma') \phi(\Gamma''))  \ \ \text{ and } \ \
  \phi_+(\Gamma)=(1-T)(\phi(\Gamma) +\sum \phi_-(\Gamma') \phi(\Gamma''))
\end{equation}
where $\Delta(\Gamma)=1\otimes \Gamma + \Gamma \otimes 1 + \sum \Gamma'\otimes \Gamma''$.

\smallskip

The Birkhoff factorization \eqref{recBirkhoff} of algebra homomorphisms 
$\phi \in {\rm Hom}_{\rm Alg}(\cH,\cR)$ is often referred to as ``algebraic
Birkhoff factorization", to distinguish it from the (analytic) Birkhoff factorization
formulated in terms of loops (or infinitesimal loops) with values in Lie
groups. We refer the reader to \S 6.4 of Chapter 1 of \cite{CoMa} for
a discussion of the relation between these two kinds of Birkhoff factorization.

\smallskip

In the original Connes--Kreimer formulation, this approach is
applied to the unrenormalized Feynman amplitudes regularized
by dimensional regularization, with the Rota--Baxter algebra
consisting of germs of meromorphic functions at the origin,
with the operator of projection onto the polar part of the
Laurent series.

\smallskip

In the following, we consider the following variant on the Hopf algebra of
Feynman graphs.

\begin{defn}\label{Heven}
As an algebra, $\cH_{\rm even}$ is
the commutative algebra generated by Feynman graphs
of a given scalar quantum field theory that have an even
number of internal edges, $\# E(\Gamma)\in 2\N$. The
coproduct \eqref{CKcoprod} on $\cH_{\rm even}$ is similarly
defined with the sum over divergent subgraphs $\gamma$ with
even $\# E(\gamma)$, with 1PI quotient.
\end{defn}

\smallskip

Notice that in dimension $D\in 4\N$ all the log divergent subgraphs
$\gamma \subset \Gamma$ have an even number of edges,
since $D b_1(\gamma) = 2 \#E(\gamma)$ in this case. This is a class
of graphs that are especially interesting in physical applications. 

\smallskip

\begin{ques}\label{grquestion}
Is there a graded-commutative version of Birkhoff factorization
involving graded-commutative Rota--Baxter and Hopf
algebras?
\end{ques}

Such an extension to the graded-commutative case would be
necessary to include the more general case of differential forms
of odd degree (associated to Feynman graphs with an odd
number of internal edges).

\medskip

One can approach the question above by using the general setting of \cite{EFGK1}:
\begin{enumerate}
\item
Let $\cH$ be any connected filtered cograded Hopf algebra and let $\cR$ be a (not necessarily commutative) associative algebra equipped with a Rota-Baxter operator of weight $\lambda\neq0$. The algebraic Birkhoff factorization of any $\phi\in{\rm Hom}(\cH,\cR)$ was obtained by Ebrahimi-Fard, Guo and Kreimer in \cite{EFGK1}.
\item
However, if the target algebra $\cR$ is not commutative, the set of characters $\Hom(\cH,\cR)$ is not a group since it is not closed under convolution product, i.e. if $f,g\in \Hom(\cH,\cR)$, then $f\star g$ does not necessarily belong  to $\Hom(\cH,\cR)$.
\end{enumerate}

The usual proof (see Theorem 4 of \cite{CoKr} and Theorem 1.39 in Chapter 1 of \cite{CoMa}) 
of the fact that the two parts $\phi_{\pm}$ of the Birkhoff factorization are algebra homomorphisms
uses explicitly both the commutativity of the target Rota--Baxter algebra $\cR$ and the fact that
${\rm Hom}_{\rm Alg}(\cH,\cR)$ is a group, and does not extend directly to the 
graded-commutative case. The argument given in Theorems 3.4 and 3.7 of \cite{EFGK1}
provides a more general form of Birkhoff factorization that applies to a graded-commutative (and
more generally non-commutative) Rota--Baxter algebra. The resulting form of the 
factorization is more complicated than in the commutative case, in general. However,
if the Rota--Baxter operator of weight $-1$ also satisfies $T^2=T$ and  
$T(T(x)y)=T(x)y$ for all $x,y\in \cR$, then the form of the Birkhoff factorization for 
not necessarily commutative Rota--Baxter algebras simplifies considerably, 
and the $\phi_+$ part of the factorization consists of a simple pole subtraction, as we
prove in Proposition \ref{phiproperties} below. 

\subsection{Rota--Baxter algebras and Atkinson factorization}\label{Atk:Sec}

In the following we will discuss some interesting properties of algebraic Birkhoff decomposition when the Rota-Baxter operator satisfies the identity
$T(T(x)y)=T(x)y$.

Let $e:\cH\rightarrow \cR$ be the unit of ${\rm Hom}(\cH,\cR)$ (under the convolution product) defined by $e({\rm 1}_\cH)={\rm 1}_\cR$ and $e(\Gamma)=0$ on $\oplus_{n>0}\cH_n$.

The main observation can be summarized as follows:
\begin{enumerate}
\item
If the Rota-Baxter operator $T$ on $\cR$ also satisfies the identity $T(T(x)y)=T(x)y$, then on ${\rm ker}(e)=\oplus_{n>0}\cH_n$,
the negative part of the Birkhoff factorization $\phi_{-}$ takes the following form:
$$\phi_{-}=-T(\phi(\Gamma))-\sum T(\phi(\Gamma'))\phi(\Gamma''),\quad {\rm for}\; \Delta(\Gamma)=1\otimes \Gamma+\Gamma\otimes 1+\sum \Gamma'\otimes \Gamma''.$$
\item If $T$ also satisfies
$T(xT(y))=xT(y)$, $\forall x,y\in \cR$,
then the positive part is given by $\phi_{+}(\Gamma)=(1-T)(\phi(\Gamma))$, $\forall \Gamma\in {\rm ker}(e)=\oplus_{n>0}\cH_n$.
\end{enumerate}

This follows from the properties of the Atkinson Factorization in Rota--Baxter algebras, which we recall below.

\begin{prop} {\rm (Atkinson Factorization, \cite{Atkin}, see also \cite{Guo1})}  Let $(\cR,T)$ be a Rota-Baxter algebra of weight $\lambda\neq0$. Let $\tilde{T}=-\lambda{\rm id}-T$ and let $a\in \cR$. Assume that $b_l$ and $b_r$ are solutions of the fixed point equations
\begin{equation}\label{fixed}
b_l=1+T(b_la),\quad b_r=1+\tilde{T}(ab_r).
\end{equation}
Then
$$b_l(1+\lambda a)b_r=1.$$
Thus
\begin{equation}\label{factor}
1+\lambda a= b_l^{-1}b_r^{-1}
\end{equation}
if $b_l$ and $b_r$ are invertible.
\end{prop}
A Rota-Baxter algebra $(\cR,T)$ is called complete if there are algebras $\cR_n\subseteq \cR, n\geq 0$, such that $(\cR,\cR_n)$ is a complete algebra and $T(\cR_n)\subseteq \cR_n$.

\begin{prop} {\rm (Existence and uniqueness of the Atkinson Factorization, \cite{Guo1})} \label{unique}
Let $(\cR,T,\cR_n)$ be a complete Rota-Baxter algebra of weight $\lambda\neq0$. Let $\tilde{T}=-\lambda{\rm id}-T$ and let $a\in \cR_1$.
\begin{enumerate}
\item
Equations (\ref{fixed}) have unique solutions $b_l$ and $b_r$. Further $b_l$ and $b_r$ are invertible. Hence the Atkinson Factorization (\ref{factor}) exists.
\item
If $\lambda\neq 0$ and $T^2=-\lambda T$ (in particular if $T^2=-\lambda T$ on $\cR$), then there are unique $c_l\in 1+T(\cR)$ and $c_r\in 1+\tilde{T}(\cR)$ such that $$1+\lambda a=c_lc_r.$$
\end{enumerate}
\end{prop}
Define
$$(Ta)^{[n+1]}:=T((Ta)^{[n]}a)\quad {\rm and}\quad (Ta)^{\{n+1\}}=T(a(Ta)^{\{n\}})$$
with the convention that $(Ta)^{[1]}=T(a)=(Ta)^{\{1\}}$ and $(Ta)^{[0]}=1=(Ta)^{\{0\}}$.

\begin{prop}\label{simplified}
Let $(\cR,\cR_n,T)$ be a complete filtered Rota-Baxter algebra of weight $-1$ such that
$T^2=T$. Let $a\in \cR_1$.
If $T$ also satisfies the following identity
\begin{equation}
T(T(x)y)=T(x)y,\quad \forall x,y\in \cR,\label{Tproperty}
\end{equation}
then the equation
\begin{equation}
b_l=1+T(b_la).\label{fixed1}
\end{equation}
has a unique solution
$$1+T(a)(1-a)^{-1}.$$\label{item1}
\end{prop}

\begin{proof}
First, we have $(Ta)^{[n+1]}=T(a)a^n$ for $n\geq0$. In fact, the case when $n=0$ just follows from the definition. Suppose it is true up to $n$, then $(Ta)^{[n+2]}=T((Ta)^{[n+1]}a)=T((T(a)a^n)a)=T(T(a)a^{n+1})=T(a)a^{n+1}$. Arguing as in \cite{EFGK1}, $b_l=\sum_{n=0}^{\infty}(Ta)^{[n]}=1+T(a)+T(T(a)a)+\cdot\cdot\cdot+(Ta)^{[n]}+\cdot\cdot\cdot$ is the unique solution of \eqref{fixed1}. So
\begin{eqnarray*}
b_l&=&1+T(a)+T(a)a+T(a)a^2+\cdot\cdot\cdot\\
&=&1+T(a)(1+a+a^2+\cdot\cdot\cdot)\\
&=&1+T(a)(1-a)^{-1}.
\end{eqnarray*}
\end{proof}

A bialgebra $\cH$ over a field $K$ is called a connected, filtered cograded bialgebra if there are subspaces $\cH_n$ of $\cH$ such that $(a)$ $\cH_p\cH_q\subseteq\sum_{k\leq p+q}\cH_k$; $(b)$ $\Delta(\cH_n)\subseteq\oplus_{p+q=n}\cH_p\otimes \cH_q$; $(c)$ $\cH_0={\rm im} (u)=K$, where $u:K
\rightarrow \cH$ is the unit of $\cH$.

\begin{prop}\label{phiproperties}
Let $\cH$ be a connected filtered cograded bialgebra (hence a Hopf algebra) and let $(\cR,T)$ be a (not necessarily commutative) Rota-Baxter algebra of weight $\lambda=-1$ with $T^2=T$. Suppose that $T$ also satisfies (\ref{Tproperty}). Let $\phi:\cH\rightarrow \cR$ be a character, that is, 
an algebra homomorphism. Then
there are unique maps $\phi_{-}:\cH\rightarrow T(\cR)$ and $\phi_{+}:\cH\rightarrow \tilde{T}(\cR)$,
where $\tilde T=1-T$, such that
$$\phi=\phi_{-}^{\ast(-1)}\ast\phi_{+}, $$
where $\phi^{\ast(-1)}=\phi\circ S$, with $S$ the antipode.
$\phi_{-}$ takes the following form on ${\rm ker}(e)=\oplus_{n>0}\cH_n$:
\begin{eqnarray*}
\phi_{-}(\Gamma)&=&-T(\phi(\Gamma))-\sum_{n=1}^\infty(-1)^n\sum T(\phi(\Gamma^{(1)}))\phi(\Gamma^{(2)})\phi(\Gamma^{(3)})\cdot\cdot\cdot \phi(\Gamma^{(n+1)})\\
&=&-T(\phi(\Gamma))-\sum_{n=1}^\infty(-1)^n((T\phi)\tilde{\ast}\phi^{\tilde{\ast}^n})(\Gamma).
\end{eqnarray*}
Here we use the notation $\tilde{\Delta}^{n-1}(\Gamma)=\sum \Gamma^{(1)}\otimes\cdot\cdot\cdot\otimes \Gamma^{(n)}$, and $\tilde{\Delta}(\Gamma):=\Delta(\Gamma)-\Gamma\otimes1-1\otimes \Gamma$ (which is coassociative), and $\tilde{\ast}$ is the convolution product defined by $\tilde{\Delta}$.
Furthermore, if $T$  satisfies
\begin{equation}
T(xT(y))=xT(y),\quad \forall x,y\in A,\label{2property}
\end{equation}
then $\phi_{+}$ takes the form on ${\rm ker}(e)=\oplus_{n>0}\cH_n$:
$$\phi_{+}(\Gamma)=(1-T)(\phi(\Gamma)).$$
\end{prop}

\begin{proof}
Define $R:={\rm Hom}(\cH,\cR)$ and
$$P:R\rightarrow R,\quad P(f)(\Gamma)=T(f(\Gamma)),\; f\in{\rm Hom}(\cH,\cR), \Gamma\in \cH.$$
Then by \cite{Guo1}, $R$ is a complete algebra with filtration $R_n=\{f\in{\rm Hom}(\cH,\cR)|f(\cH_{n-1})=0\}, n\geq 0$, and $P$ is a Rota-Baxter operator of weight $-1$ and $P^2=P$. Moreover, since $T$ satisfies (\ref{Tproperty}), it is easy to check that $P(P(f)g)=P(f)g$ for any $f,g\in{\rm Hom}(\cH,\cR)$.
Let $\phi:\cH\rightarrow \cR$ be a character. Then $(e-\phi)({\rm 1}_\cH)=e({\rm 1}_{\cH})-\phi({\rm 1}_{\cH})={\rm 1}_\cR-{\rm 1}_\cR=0$. So $e-\phi\in \cR_1$. Set $a=e-\phi$, by Proposition \ref{unique}, we know that there are unique $c_l\in T(\cR)$ and $c_r\in (1-T)(\cR)$ such that $\phi=c_lc_r$. Moreover, by Proposition \ref{simplified} we have $\phi_{-}=b_l=c_l^{-1}=e+T(a)(e-a)^{-1}=e+T(e-\phi)\sum_{n=0}^\infty(e-\phi)^n$.
We also have $\sum_{n=0}^\infty(e-\phi)^n({\rm 1}_\cH)={\rm 1}_\cR$ and for any $X\in{\rm ker}(e)=\oplus_{n>0}\cH_n$, we have $(e-\phi)^0(\Gamma)=e(\Gamma)=0$; $(e-\phi)^1(\Gamma)=-\phi(\Gamma)$; $(e-\phi)^2(\Gamma)=\sum(e-\phi)(\Gamma')(e-\phi)(\Gamma'')=\sum\phi(\Gamma')\phi(\Gamma'')$. More generally, we have $(e-\phi)^n(\Gamma)=(-1)^n\sum\phi(\Gamma^{(1)})\phi(\Gamma^{(2)})\cdots\phi(\Gamma^{(n)})=(-1)^n\phi^{\tilde{\ast}^n}(\Gamma)$.
So for $X\in{\rm ker}(e)=\oplus_{n>0}\cH_n$,
\begin{eqnarray*}
\phi_{-}(\Gamma)&=&(T(e-\phi)\sum_{n=0}^\infty(e-\phi)^n)(\Gamma)\\
&=&T(e-\phi)({\rm 1}_\cH)\sum_{n=0}^\infty(e-\phi)^n(\Gamma)+T(e-\phi)(\Gamma)\sum_{n=0}^\infty(e-\phi)^n({\rm 1}_\cH)\\
& &+\sum T((e-\phi)(\Gamma'))\sum_{n=1}^\infty(e-\phi)^n(\Gamma'')\\
&=&-T(\phi(\Gamma))-\sum T(\phi(\Gamma'))\sum_{n=1}^\infty(-1)^n\sum\phi((\Gamma'')^{(1)})\phi((\Gamma'')^{(2)})\cdots \phi((\Gamma'')^{(n)})\\
&=&-T(\phi(\Gamma))-\sum_{n=1}^\infty(-1)^n\sum T(\phi(\Gamma^{(1)}))\phi(\Gamma^{(2)})\phi(X^{(3)})\cdots \phi(\Gamma^{(n+1)})\\
&=&-T(\phi(\Gamma))-\sum_{n=1}^\infty(-1)^n((T\phi)\tilde{\ast}\phi^{\tilde{\ast}^n})(\Gamma).
\end{eqnarray*}
Suppose that $T$ also satisfies equation \eqref{2property}, then for any $a,b\in \cR$, we have $(1-T)(a)(1-T)(b)=ab-T(a)b-aT(b)+T(a)T(b)=ab-T(T(a)b)-T(aT(b))+T(a)T(b)=ab-T(ab)=(1-T)(ab)$, as $T$ is a Rota-Baxter operator of weight $-1$. As shown in \cite{CoKr} and \cite{EFGK1}, $\phi_{+}(\Gamma)=(1-T)(\phi(\Gamma)+\sum\phi_{-}(\Gamma')\phi(\Gamma''))$. So $\phi_{+}(\Gamma)=(1-T)(\phi(\Gamma))+\sum(1-T)(\phi_{-}(\Gamma'))(1-T)(\phi(\Gamma''))$ by the previous computation. But $\phi_{-}$ is in the image of $T$ and $T^2=T$, so we must have $(1-T)(\phi_{-}(\Gamma'))=0$, which shows that $\phi_{+}(\Gamma)=(1-T)(\phi(\Gamma))$.
\end{proof}

\subsection{A variant of algebraic renormalization}\label{renformsec}

We consider now a setting inspired by the formalism of the Connes--Kreimer
renormalization recalled above. The setting generalizes the one considered in
\cite{CeMa} for configuration space integrals and our main application will
be to extend the approach of \cite{CeMa} to momentum space integrals.

\smallskip

The main difference with respect to the Connes--Kreimer renormalization
is that, instead of renormalizing the Feynman amplitude
(regularized so that it gives a meromorphic function), we
propose to renormalize the differential form, before integration,
and then integrate the renormalized form to obtain a period.

\smallskip

The result obtained by this method differs from the physical renormalization,
as we will discuss further in Section \ref{FeynmanSec} below.
There are at present no explicit examples of periods that are known not
to be expressible in terms of rational combinations of mixed Tate periods, 
just because no such general statement of algebraic 
independence of numbers is known. However, it is generally expected that 
motives that are not mixed Tate will have periods that are
not expressible in terms of mixed Tate periods, for instance periods associated 
to $H^1$ of an elliptic curve. There are known examples (\cite{BrDo},
\cite{BrSch}) of Feynman integrals that give periods of non-mixed Tate motives
(a K3 surface, for instance). 
In our setting, the period obtained by applying the Birkhoff factorization to
the Feynman integrand $\eta_\Gamma$ is always a mixed Tate period.
However, it is difficult to ensure that the result is non-trivial. As we will
discuss in more detail in Section~\ref{KGLsec}, one can ensure a non-trivial result 
by replacing the form $\eta_\Gamma$ with a cohomologous form with logarithmic
poles and taking into account both the result of the pole subtraction and all the
Poincar\'e residues. However, passing to a form with logarithmic poles
requires, in general, restricting to the big cell of the Kausz compactification,
and this introduces a constraint on the nature of the period.  
If the intersection of the big cell of the
Kausz compactification with the divisor $\Sigma_{\ell, g}$ that contains
the boundary of the chain of integeration is a mixed Tate motive, then the
convergent integral we obtain by replacing the integration form with a
form with logarithmic poles is a mixed Tate period. For particular graphs,
for which the form with logarithmic poles extends globally to the Kausz compactification,
with poles along the boundary divisor, we obtain a mixed Tate period without
any further assumption.

\smallskip

The main steps required for our setup are the following. For a variety $X$, we denote
by $\m(X)$ the motive in the Voevodsky category. 

\begin{itemize}
\item For each $\ell\geq 1$, we construct
a pair $(X_\ell, Y_\ell)$ of a smooth projective variety $X_\ell$ (defined over $\Q$) whose
motive $\m(X_\ell)$ is mixed Tate (over $\Z$), together with a (singular)
hypersurface $Y_\ell \subset X_\ell$.
\item For each Feynman graph $\Gamma$ with loop number $\ell$ we construct a map $\Upsilon:
\A^n \smallsetminus \hat X_\Gamma \to X_\ell \smallsetminus Y_\ell$, where $\hat X_\Gamma\subset \A^n$ is the affine graph hypersurface, with $n$ the number of edges of $\Gamma$. 
\item Using the map $\Upsilon$, we describe the Feynman integrand as a morphism of commutative algebras
$$ \phi: \cH_{\rm even} \to \bigwedge_\ell \cM^{\rm even}_{X_\ell, Y_\ell}, \ \ \  \phi(\Gamma)
= \eta_\Gamma, $$
with $\cH$ the Connes--Kreimer Hopf algebra and
with the Rota--Baxter structure of Lemma \ref{prodRB} on the target algebra,
and with $\eta_\Gamma$ an algebraic differential form on $X_\ell$ with polar locus $Y_\ell$, 
for $\ell=b_1(\Gamma)$.
\item We express the (unrenormalized) Feynman
integrals as a (generally divergent) integral $\int_{\Upsilon(\sigma)} \eta_\Gamma$, over
a chain $\Upsilon(\sigma)$ in $X_\ell$ that is the image of a chain $\sigma$ in $\A^n$.
\item We construct a divisor $\Sigma_\ell\subset X_\ell$, that contains the
boundary $\partial \Upsilon(\sigma)$, whose motive $\m(\Sigma_\ell)$ is
mixed Tate (over $\Z$) for all $\ell \geq 1$.
\item We perform the Birkhoff decomposition $\phi_\pm$ obtained inductively
using the coproduct on $\cH$ and the Rota--Baxter operator $T$ (polar part) on
$\cM^*_{X_\ell, Y_\ell}$.
\item This gives a holomorphic form $\phi_+(\Gamma)$
on $X_\ell$. The divergent Feynman integral is then replaced by the integral
$$ \int_{\Upsilon(\sigma)} \phi_+(\Gamma) $$
which is a period of the mixed Tate motive $\m(X_\ell, \Sigma_\ell)$.
\item In addition to the integral of $\phi_+(\Gamma)$ on $X_\ell$ we consider
integrals on the strata of the complement $X_\ell \smallsetminus Y_\ell$ of the
polar part $\phi_-(\Gamma)$, which under suitable conditions will be
interpreted as Poincar\'e residues.
\end{itemize}

If convergent, the Feynman integral $\int_{\Upsilon(\sigma)} \eta_\Gamma$ would be a period of
$\m(X_\ell\smallsetminus Y_\ell,  \Sigma_\ell \smallsetminus (\Sigma_\ell \cap Y_\ell))$.
The renormalization procedure described above replaces it with a (convergent)
integral that is a period of the simpler motive $\m(X_\ell, \Sigma_\ell)$.
By our assumptions on $X_\ell$ and $\Sigma_\ell$, the motive $\m(X_\ell, \Sigma_\ell)$ is
mixed Tate for all $\ell$.

Thus, this strategy eliminates the difficulty of analyzing the motive
$\m(X_\ell \smallsetminus Y_\ell, \Sigma_\ell\smallsetminus (\Sigma_\ell\cap Y_\ell))$
encountered for instance in \cite{AluMa}. The form of renormalization proposed here always
produces a mixed Tate period, but at the cost of incurring in a considerable loss of information
with respect to the original Feynman integral.

\smallskip

Indeed, a difficulty in the procedure described above is ensuring that
the resulting regularized form
$$ \phi_+(\Gamma) = (1-T) (\phi(\Gamma) + \sum_{\gamma \subset \Gamma}
\phi_-(\gamma) \wedge \phi(\Gamma/\gamma)) $$
is nontrivial. This condition may be difficult to control in explicit cases,
although we will discuss below (see Section~\ref{KGLsec}) conditions under which one
can reduce the problem to forms with logarithmic poles, where using the
pole subtraction together with Poincar\'e residues one can obtain
nontrivial periods (although the result one obtains is not equivalent to
the physical renormalization of the Feynman amplitude).

\smallskip

An additional difficulty that can cause loss of information with respect to the
Feynman integral is coming from the combinatorial 
conditions on the graph given in \cite{AluMa} that we will use to ensure that the map $\Upsilon$ 
to the complement of the determinant hypersurface is an embedding, see Section \ref{FeynmanSec}.

\section{Rota--Baxter algebras and forms with logarithmic poles}\label{RBlogSec}

We now focus on the case of meromorphic forms with logarithmic poles,
where the Rota--Baxter structure and the renormalization procedure
described above drastically simplify.

\begin{lem}\label{Pnlogforms}
Let $X$ be a smooth projective variety and $Y\subset X$ a smooth hypersurface
with defining equation $Y=\{ f=0 \}$.
Let $\Omega^\star_{X}(\log(Y))$ be the sheaf of algebraic differential
forms on $X$ with logarithmic poles along $Y$. After passing to global sections,
we obtain a graded-commutative algebra, which we still denote by 
$\Omega^\star_{X}(\log(Y))$, for simplicity. 
The Rota--Baxter operator $T$ of Lemma \ref{RBmerom}
preserves the commutative subalgebra $\Omega^{\rm even}_{X}(\log(Y))$ and the pair
$(\Omega^{\rm even}_{X}(\log(Y)),T)$
is a graded Rota--Baxter algebra of degree $-1$ with the property that, for all
$\omega_1,\omega_2\in \Omega^{\rm even}_{X}(\log(Y))$, the wedge product
$T(\omega_1)\wedge T(\omega_2)=0$.
\end{lem}

\proof Forms $\omega \in \Omega^\star_{X}(\log(Y))$ can be written in canonical
form
$$ \omega=\frac{df}{f} \wedge \xi + \eta, $$
with $\xi$ and $\eta$ holomorphic, so that $T(\omega)=\frac{df}{f} \wedge \xi$.
We then have \eqref{RBmin1} as in Remark \ref{RBomegaY} above, with
$T(\omega_1) \wedge T(\omega_2)=(-1)^{|\xi_1|+1} \alpha \wedge \alpha \wedge \xi_1 \wedge \xi_2$
where $\alpha$ is the 1-form $\alpha=df/f$ so that $\alpha\wedge\alpha=0$.
\endproof

Lemma \ref{Pnlogforms} shows that, when restricted to
$\Omega^\star_{X}(\log(Y))$, the operator $T$ satisfies the simpler
identity
\begin{equation}\label{nonRB}
T(xy) = T(T(x)y) + T(x T(y)).
\end{equation}
This property greatly simplifies the decomposition of the algebra induced by the Rota--Baxter
operator. 

Let $\cR_+=(1-T)\cR$. For an operator $T$ satisfying \eqref{nonRB} and $T(x)T(y)=0$,
for all $x,y \in \cR$, the property that  $\cR_+\subset \cR$ is a subalgebra follows immediately
from the simple identity $$(1-T)(x)\cdot (1-T)(y)=xy -T(x)y - x T(y) $$ $$ = xy -T(x)y -xT(y)
-( T(xy) - T(T(x)y) - T(x T(y))) = (1-T) (xy -T(x)y -xT(y)).$$
Moreover, we obtain a simplified form of the general result of Proposition \ref{phiproperties},
when taking into account the vanishing $T(x)T(y)=0$, as shown in Lemma \ref{Pnlogforms}.

\begin{lem}\label{RBsimple}
Let $\cR$ be a commutative algebra and $T: \cR \to \cR$ a linear operator
that satisfies the identity \eqref{nonRB} and such that,
for all $x,y \in \cR$, the product $T(x) T(y)=0$.
Then both $T$ and $1-T$ are idempotent, $T^2=T$ and $(1-T)^2=1-T$.
\end{lem}

\proof The identity \eqref{nonRB} gives $T(1)=0$, since taking $x=y=1$
one obtains $T(1)=2 T^2(1)$ while taking $x=T(1)$ and $y=1$ gives
$T^2(1)=T^3(1)$. Then \eqref{nonRB} with $y=1$ gives
$T(x)=T(x T(1))+T(T(x) 1)=T^2(x)$ for all $x\in \cR$. For $1-T$ we
then have $(1-T)^2(x)=x-2T(x)+T^2(x)=(1-T)(x)$, for all $x\in \cR$.
\endproof

\begin{lem}\label{derRB}
Let $\cR$ be a commutative algebra and $T: \cR \to \cR$ a linear operator
that satisfies the identity \eqref{nonRB} and such that,
for all $x,y \in \cR$ the product $T(x) T(y)=0$. If, for all $x,y \in \cR$,
the identity $T(x)y+x T(y)= T(T(x)y)+T(x T(y))$ holds, then
the operator $(1-T): \cR \to \cR_+$ is an algebra homomorphism
and the operator $T$ is a derivation on $\cR$.
\end{lem}

\proof We have $(1-T)(xy)= xy - T(T(x)y) -T(x T(y))$ while
$(1-T)(x) \cdot (1-T)(y)=xy -T(x)y -xT(y)$. Assuming that, for all $x,y\in \cR$,
we have $T(T(x)y) +T(x T(y))=T(x)y +xT(y)$ gives $(1-T)(xy)=(1-T)(x)\cdot (1-T)(y)$.
Moreover, the identity \eqref{nonRB} can be rewritten as
$T(xy)=T(x)y+x T(y)$, hence $T$ is just a derivation on $\cR$.
\endproof

\smallskip

Consider then the case of a smooth hypersurface $Y$ in a smooth projective variety
$X$. We have the following properties.

\begin{prop}\label{logRBder}
Let $Y\subset X$ be a smooth hypersurface in a smooth projective variety.
The Rota--Baxter operator $T: \cM^{\rm even}_{X,Y} \to \cM^{\rm even}_{X,Y}$ of weight $-1$
on meromorphic forms on $X$ with poles along $Y$ restricts to a derivation on the graded
algebra $\Omega^{\rm even}_{X}(\log(Y))$ of forms with logarithmic poles. Moreover,
the operator $1-T$ is a morphism of commutative algebras from
$\Omega^{\rm even}_{X}(\log(Y))$ to the algebra of holomorphic forms
$\Omega^{\rm even}_{X}$.
\end{prop}

\proof It suffices to check that the polar part operator $T: \Omega^{\rm even}_{X}(\log(Y)) \to
\Omega^{\rm even}_{X}(\log(Y))$ satisfies the hypotheses of Lemma \ref{derRB}. We have
seen that, for all $\omega_1, \omega_2\in \Omega^{\rm even}_{X}(\log(Y))$, the product
$T(\omega_1)\wedge T(\omega_2)=0$. Moreover, for $\omega_i = d\log(f) \wedge \xi_i + \eta_i$,
we have $T(\omega_1)\wedge \omega_2 = d\log(f)\wedge \xi_1 \wedge \eta_2$ and
$\omega_1 \wedge T(\omega_2)= (-1)^{|\eta_1|} d\log(f) \wedge \eta_1 \wedge \xi_2$,
where the $\xi_i$ and $\eta_i$ are holomorphic, so that we have $T( T(\omega_1)\wedge \omega_2 )=
T(\omega_1)\wedge \omega_2$ and $T(\omega_1 \wedge T(\omega_2))=\omega_1 \wedge T(\omega_2)$.
Thus, the hypotheses of Lemma \ref{derRB} are satisfied.
\endproof

\subsection{Birkhoff factorization and forms with logarithmic poles}

In cases where the pair $(X, Y)$ has the property that
all de Rham cohomology classes in $H^*_{dR}(X \smallsetminus Y)$
are represented by global algebraic differential forms with logarithmic poles,
the construction above simplifies significantly. Indeed, the Birkhoff
factorization becomes essentially trivial, because of Proposition
\ref{logRBder}. In other words, all graphs behave ``as if they were
log divergent". This can be stated more precisely as follows.

\smallskip

\begin{prop}\label{logpolesT}
Let $Y \subset X$ be a smooth hypersurface inside a smooth projective variety and let
$\Omega^{\rm even}_{X}(\log(Y))$ denote the commutative
algebra of algebraic differential forms on $X$ of even degree with logarithmic
poles on $Y$. Let $\phi: \cH \to \Omega^{\rm even}_{X}(\log(Y))$
be a morphism of commutative algebras from a commutative
Hopf algebra $\cH$ to
$\Omega^{\rm even}_{X}(\log(Y))$ with the operator $T$ of pole
subtraction. Then for every $\Gamma\in \cH$ one has
$$ \phi_+(\Gamma) =(1-T) \phi(\Gamma), $$
while the negative part of the Birkhoff factorization takes the form
$$ \phi_-(\Gamma) = -T (\phi(\Gamma)) - \sum \phi_-(\Gamma') \phi(\Gamma''), $$
where $\Delta(\Gamma)=\Gamma\otimes 1 + 1 \otimes \Gamma + \sum \Gamma' \otimes \Gamma''$.
Moreover, $\phi_{-}$ takes the following nonrecursive form on ${\rm ker}(e)=\oplus_{n>0}\cH_n$:
\begin{eqnarray*}
\phi_{-}(\Gamma)&=&-T(\phi(\Gamma))-\sum_{n=1}^\infty(-1)^n\sum T(\phi(\Gamma^{(1)}))\phi(\Gamma^{(2)})\phi(\Gamma^{(3)})\cdot\cdot\cdot \phi(\Gamma^{(n+1)})\\
&=&-T(\phi(\Gamma))-\sum_{n=1}^\infty(-1)^n((T\phi)\tilde{\ast}\phi^{\tilde{\ast}^n})(\Gamma).
\end{eqnarray*}
\end{prop}

\proof The operator $T$ of pole subtraction is a derivation
on $\Omega^{\rm even}_{X}(\log(Y))$. By \eqref{recBirkhoff} we have
$\phi_+(\Gamma)= (1-T) (\phi(\Gamma) + \sum \phi_-(\Gamma') \phi(\Gamma''))$. By
Proposition \ref{logRBder} we know that, in the case of forms
with logarithmic poles along a smooth hypersurface, $1-T$ is an algebra homomorphism,
hence $\phi_+(\Gamma)= (1-T) (\phi(\Gamma)) + \sum (1-T)(\phi_-(\Gamma')) (1-T)(\phi(\Gamma'')))$,
but $\phi_-(\Gamma')$ is in the range of $T$ and, again by Proposition \ref{logRBder},
we have $T^2=T$, so that the terms in the sum all vanish, since $(1-T)(\phi_-(\Gamma')) =0$.
By \eqref{recBirkhoff} we have
$\phi_-(\Gamma)=-T (\phi(\Gamma) + \sum \phi_-(\Gamma') \phi(\Gamma'')) = -T \phi(\Gamma) - \sum T(\phi_-(\Gamma')) \phi(\Gamma'')
-\sum \phi_-(\Gamma') T(\phi(X\Gamma'))$, because by Proposition \ref{logRBder} $T$ is a derivation.
The last sum vanishes because $\phi_-(\Gamma')$ is in the range of $T$ and we have
$T(\eta)\wedge T(\xi)=0$ for all $\eta,\xi \in \Omega^*_{X}(\log(Y))$. Thus,
we are left with $\phi_-(\Gamma)=-T \phi(\Gamma) - \sum T(\phi_-(\Gamma')) \phi(\Gamma'') =-T \phi(\Gamma) - \sum
\phi_-(\Gamma') \phi(\Gamma'')$. The last part follows from Proposition \ref{phiproperties},
since $T(T(\eta)\wedge\xi)=T(\eta)\wedge\xi$.
\endproof

\smallskip

Notice that this is compatible with the property that
$\phi(\Gamma) =(\phi_-\circ S \star \phi_+)(\Gamma)$ (with the $\star$-product
dual to the Hopf algebra coproduct). In fact, this identity is equivalent
to $\phi_+=\phi_-\star \phi$, which means that $\phi_+(\Gamma)=\langle
\phi_-\otimes \phi , \Delta(\Gamma)\rangle = \phi_-(\Gamma) + \phi(\Gamma) +\sum
\phi_-(\Gamma') \phi(\Gamma'') =(1-T) \tilde\phi(\Gamma)$ as above. Equivalently, all
the nontrivial terms $\phi_-(\Gamma') \phi(\Gamma'')$ in $\tilde\phi(\Gamma)$ satisfy
$T(\phi_-(\Gamma') \phi(\Gamma''))=\phi_-(\Gamma') \phi(\Gamma'')$, because of the simplified
form \eqref{nonRB} of the Rota--Baxter identity.

\medskip

\begin{cor}\label{OneXYcor}
Suppose given a character $\phi: \cH \to \Omega^{\rm even}_X(\log(Y))$ 
of the Hopf algebra of Feynman graphs, where $X=X_\ell$ and $Y=Y_\ell$ independently 
of the number of loops $\ell\geq 1$. Then the negative part of the Birkhoff factorization
of Proposition \ref{logpolesT} has the simple form
\begin{equation}\label{phiminXY}
\phi_-(\Gamma)=- \frac{dh}{h} \wedge \left( \xi_\Gamma +
\sum_{N\geq 1} (-1)^N\sum_{\gamma_N\subset \cdots \subset \gamma_1 \subset \gamma_0=\Gamma}
\xi_{\gamma_N} \wedge \bigwedge_{j=1}^N \eta_{\gamma_{j-1}/\gamma_j} \right),
\end{equation}
where $\phi(\Gamma)=\frac{dh}{h} \wedge \xi_\Gamma + \eta_\Gamma$, and $Y=\{ h=0 \}$.
\end{cor}

\proof The result follows from the expression
$$ \phi_-(\Gamma) = -T (\phi(\Gamma)) - \sum_{\gamma \subset \Gamma} \phi_-(\gamma) \phi(\Gamma/\gamma), $$
obtained in Proposition \ref{logpolesT}, where $\phi(\Gamma)=\omega_\Gamma = \frac{dh}{h} \wedge \xi_\Gamma + \eta_\Gamma$, so that $T(\phi(\Gamma))=\frac{dh}{h} \wedge \xi_\Gamma$ and $\phi(\Gamma/\gamma) = \frac{dh}{h} \wedge \xi_{\Gamma/\gamma} + \eta_{\Gamma/\gamma}$. The wedge product of
$\phi_-(\gamma) = -T(\phi(\gamma)) - \sum_{\gamma_2\subset \gamma} \phi_-(\gamma_2) \phi(\gamma/\gamma_2)$
with $\phi(\Gamma/\gamma)$ will give a term $\frac{dh}{h} \wedge \xi_{\gamma} \wedge \eta_{\Gamma/\gamma}$
and additional terms $\phi_-(\gamma_2) \phi(\gamma/\gamma_2) \wedge \eta_{\Gamma/\gamma}$. Proceeding
inductively on these terms, one obtains \eqref{phiminXY}.
\endproof

\medskip

\begin{rem}\label{notoneXY}{\rm 
In the geometric construction we consider here, one does not have a single pair $(X,Y)$ for
all loop numbers. Instead, we consider a more general situation, where $X_\ell$ and $Y_\ell$ 
depend on the loop number $\ell \geq 1$. In this case, the form of the negative piece 
$\phi_-(\Gamma)$ is more complicated than in Corollary \ref{OneXYcor}, as it contains forms
on the products $X_{\ell(\gamma)}\times X_{\ell(\Gamma/\gamma)}$ with logarithmic poles
along $Y_{\ell(\gamma)}\times X_{\ell(\Gamma/\gamma)} \cup X_{\ell(\gamma)}\times Y_{\ell(\Gamma/\gamma)}$. However, the general form of the expression is similar, only more cumbersome to
write explicitly.}
\end{rem}

\smallskip
\subsection{Polar subtraction and the residue}\label{ResSec}

We have seen that, in the case of a smooth hypersurface $Y\subset X$,
the Birkhoff factorization in the algebra of forms with logarithmic poles
reduces to a simple pole subtraction, $\phi_+(\Gamma)=(1-T)\phi(\Gamma)$. If
the unrenormalized $\phi(\Gamma)$ is a form written as $\alpha + \frac{df}{f}\wedge \beta$,
with $\alpha$ and $\beta$ holomorphic, then $\phi_+(\Gamma)$ vanishes identically
whenever $\alpha=0$. In that case, all information about $\phi(\Gamma)$ is lost in
the process of pole substraction. Suppose that $\int_\sigma \phi(\Gamma)$ is the
original unrenormalized integral. To maintain some additional information, it
is preferable to consider, in addition to the integral $\int_\sigma \phi_+(\Gamma)$,
also an integral of the form
$$ \int_{\sigma \cap Y} {\rm Res}_Y(\eta), $$
where ${\rm Res}_Y(\eta)=\beta$ is the Poincar\'e residue of
$\eta=\alpha + \frac{df}{f}\wedge \beta$ along $Y$. It is dual to
the Leray coboundary, in the sense that
$$ \int_{\sigma\cap Y} {\rm Res}_Y(\eta) =\frac{1}{2\pi i} \int_{\cL(\sigma\cap Y)} \eta, $$
where the Leray coboundary $\cL(\sigma\cap Y)$ is a circle bundle over
$\sigma\cap Y$. In this way, even when $\alpha=0$, one can still retain the
nontrivial information coming from the Poincar\'e residue, which is
also expressed as a period.

\smallskip
\section{Singular hypersurfaces and meromorphic forms}\label{SingHypSec}

In our main application, we will need to work with pairs $(X,Y)$ where
$X$ is smooth projective, but the hypersurface $Y$ is singular. Thus,
we now discuss extensions of the results above to more general
situations where $Y\subset X$ is a singular hypersurface in a
smooth projective variety $X$.

\smallskip

Again we denote by $\cM^*_{X,Y}$ the sheaf of meromorphic differential forms on $X$
with poles along $Y$, of arbitrary order, and by $\Omega_X^*(\log(Y))$ the sub-sheaf
of forms with logarithmic poles along $Y$.
Let $h$ be a local determination of $Y$, so that $Y=\{ h=0 \}$. We can then locally
represent forms $\omega \in \cM^*_{X,Y}$ as finite sums
$\omega =\sum_{p\geq 0} \omega_p / h^p$, with the $\omega_p$ holomorphic.
The polar part operator $T: \cM^{\rm even}_{X,Y} \to \cM^{\rm even}_{X,Y}$
can then be defined as in \eqref{Tmeromf}.

\smallskip

In the case we considered in the previous section, with $Y\subset X$
a smooth hypersurface, forms with logarithmic poles
can be represented as
\begin{equation}\label{YsmoothRes}
 \omega = \frac{dh}{h} \wedge \xi +\eta,
\end{equation}
with $\xi$ and $\eta$ holomorphic. The Leray residue
is given by ${\rm Res}(\omega)=\xi$. It is well defined, as the restriction of $\xi$
to $Y$ is independent of the choice of a local equation for $Y$.

\smallskip

In the next subsection we discuss how the case of a smooth hypersurface 
generalizes to the case of a normal crossings divisor $Y\subset X$ inside 
a smooth projective variety $X$. The normal crossings divisor is a particularly
nice case of a larger class of singular hypersurfaces.  
The complex of forms with logarithmic poles extends from the smooth hypersurface case
to the normal crossings divisor case as in \cite{Del}. 
For more general singular hypersurfaces, an appropriate notion of
forms with logarithmic poles was introduced by Saito in \cite{Sa}.
The construction of the residue was also generalized from the smooth hypersurface
case to the case where $Y$ is a normal crossings divisor in \cite{Del} and to more 
general singular hypersurfaces in \cite{Sa}.

\subsection{Normal crossings divisors}

The main case of singular hypersurfaces that we focus on for our applications
will be simple normal crossings divisors. In fact, while our formulation of the Feynman
amplitude in momentum space is based on the formulation of \cite{AluMa},
where the unrenormalized Feynman integral lives on the complement of the
determinant hypersurface, which has worse singularities, we will reformulate
the integral on the Kausz compactification of $\GL_n$ where the boundary
divisor of the compactification is normal crossings.

\smallskip

If $Y\subset X$ is a simple normal crossings divisor in a smooth projective variety, with
$Y_j$ the components of $Y$, with local equations $Y_j=\{ f_j=0 \}$,
the complex of forms with logarithmic poles
$\Omega^*_X(\log(Y))$ is spanned by the forms $\frac{df_j}{f_j}$ and
by the holomorphic forms on $X$.

\smallskip

As in Theorem 6.3 of \cite{CeMa}, we obtain that
the Rota--Baxter operator of polar projection
$T : \cM^{\rm even}_{X,Y}\to \cM^{\rm even}_{X,Y}$ restricts to a
Rota--Baxter operator $T: \Omega^{\rm even}_X(\log(Y))
\to \Omega^{\rm even}_X(\log(Y))$ given by
\begin{equation}\label{RBlogDiv}
T: \eta \mapsto T(\eta)= \sum_j \frac{df_j}{f_j} \wedge {\rm Res}_{Y_j}(\eta),
\end{equation}
where the holomorphic form
${\rm Res}_{Y_j}(\eta)$ is the Poincar\'e residue of $\eta$ restricted to $Y_j$.

\smallskip

Unlike the case of a single smooth hypersurface, for a simple normal
crossings divisor the Rota--Baxter operator operator $T$ does not
satisfy $T(x)T(y)\equiv 0$, since we now have terms like
$\frac{df_j}{f_j} \wedge \frac{df_k}{f_k}\neq 0$, for $j\neq k$, so the
Rota--Baxter identity for $T$ does not reduce to a derivation, but some
of the properties that simplify the Birkhoff factorization in the case of
a smooth hypersurface still hold in this case.

\begin{prop}\label{Tlogid}
The Rota--Baxter operator $T$ of \eqref{RBlogDiv} satisfies $T^2=T$ and
the Rota--Baxter identity simplifies to the form
\begin{equation}\label{simplRB}
T(\eta \wedge \xi) = T(\eta) \wedge \xi + \eta \wedge T(\xi) - T(\eta)\wedge T(\xi).
\end{equation}
The operator $(1-T): \cR \to \cR_+$ is an algebra homomorphism, with
$\cR=\Omega^{\rm even}_X(\log(Y))$ and $\cR_+=(1-T)\cR$. The
Birkhoff factorization of a commutative algebra homomorphism $\phi:\cH \to \cR$,
with $\cH$ a commutative Hopf algebra, is given by
\begin{equation}\label{BirkDivLog}
\begin{array}{l}
\phi_+(\Gamma)= (1-T) \phi(\Gamma) \\
\phi_-(\Gamma) =-T (\phi(\Gamma) + \sum \phi_-(\Gamma') \phi(\Gamma'')) .
\end{array}
\end{equation}
Moreover, $\phi_{-}$ takes the following form on
${\rm ker}(e)=\oplus_{n>0}\cH_n$:
\begin{eqnarray*}
\phi_{-}(\Gamma)&=&-T(\phi(\Gamma))-\sum_{n=1}^\infty(-1)^n\sum T(\phi(\Gamma^{(1)}))\phi(\Gamma^{(2)})\phi(\Gamma^{(3)})\cdots \phi(\Gamma^{(n+1)})\\
&=&-T(\phi(\Gamma))-\sum_{n=1}^\infty(-1)^n((T\phi)\tilde{\ast}\phi^{\tilde{\ast}^n})(\Gamma).
\end{eqnarray*}
\end{prop}

\proof The argument is the same as in the proof of Theorem 6.3 in \cite{CeMa}.
It is clear by construction that $T$ is idempotent and the simplified form \eqref{simplRB}
of the Rota--Baxter identity follows by observing that
$T(T(\eta)\wedge \xi)=T(\eta)\wedge \xi$ and $T(\eta \wedge T(\xi))=\eta\wedge T(\xi)$
as in Theorem 6.3 in \cite{CeMa}. Then one sees that $$(1-T)(\eta)\wedge
(1-T)(\xi)= \eta\wedge \xi -T(\eta)\wedge \xi - \eta \wedge T(\xi) +T(\eta)\wedge T(\xi)=
\eta\wedge \xi - T(\eta \wedge \xi)$$ by \eqref{simplRB}.
Consider then the Birkhoff factorization. We write $\tilde\phi(\Gamma):=
\phi(\Gamma) + \sum \phi_-(\Gamma') \phi(\Gamma'')$. The fact that $(1-T)$ is an
algebra homomorphism then gives
$$\phi_+(\Gamma)= (1-T)(\tilde\phi(\Gamma))= (1-T) (\phi(\Gamma) + \sum \phi_-(\Gamma') \phi(\Gamma''))
$$ $$ = (1-T)(\phi(\Gamma)) + \sum (1-T)(\phi_-(\Gamma')) (1-T)(\phi(\Gamma''))),$$ with
$(1-T)(\phi_-(\Gamma'))=-(1-T)T(\tilde\phi_-(\Gamma')) =0$, because $T$ is idempotent.
The last statement again follows from Proposition \ref{phiproperties}, since
we have $T(T(\eta)\wedge\xi)=T(\eta)\wedge\xi$.
\endproof

\smallskip

\subsection{Multidimensional residues}\label{multiresSec}

In the case of a simple normal crossings divisor $Y \subset X$,
we can proceed as discussed in Section \ref{ResSec}  for the
case of a smooth hypersurface. Indeed, as we have seen in
Proposition \ref{Tlogid}, we also have in this case a simple
pole subtraction $\phi_+(\Gamma)=(1-T)\phi(\Gamma)$, even though the
negative term $\phi_-(\Gamma)$ of the Birkhoff factorization can
now be more complicated than in the case of a smooth
hypersurface.

The unrenormalized $\phi(\Gamma)$ is a form $\eta =\alpha + \sum_j \frac{df_j}{f_j} \wedge \beta_j$,
with $\alpha$ and $\beta_j$ holomorphic and $Y_j=\{ f_j=0 \}$ the components of $Y$.
Again, if $\alpha=0$ we loose all information about $\phi(\Gamma)$ in our renormalization of the
logarithmic form. To avoid this problem, we can again consider, instead of the
single renormalized integral $\int_\sigma \phi_+(\Gamma)$, an additional family of integrals
$$ \int_{\sigma \cap Y_I} {\rm Res}_{Y_I}(\eta), $$
where $Y_I=\cap_{j\in I} Y_j$ is an intersection of components of the divisor $Y$
and ${\rm Res}_{Y_I}(\eta)$ is the iterated (or multidimensional, or higher)
Poincar\'e residue of $\eta$, in the sense of \cite{Del}. These are dual to
the iterated Leray coboundaries,
$$ \int_{\sigma \cap Y_I} {\rm Res}_{Y_I}(\eta) =\frac{1}{(2\pi i)^n} \int_{\cL_I (\sigma\cap Y_I)} \eta, $$
where $\cL_I= \cL_{j_i}\circ \cdots \circ \cL_{j_n}$ for $Y_I=Y_{j_1}\cap \cdots \cap Y_{j_n}$.

If arbitrary intersections $Y_I$ of components of $Y$ are all mixed Tate motives,
then all these integrals are also periods of mixed Tate motives.

\smallskip
\subsection{Saito's logarithmic forms}

Given a singular reduced hypersurface $Y\subset X$, a differential form $\omega$
with logarithmic poles along $Y$, in the sense of Saito \cite{Sa}, can always be written
in the form (\cite{Sa}, (1.1))
\begin{equation}\label{YsingRes}
f\, \omega = \frac{dh}{h} \wedge \xi + \eta,
\end{equation}
where $f \in \cO_X$ defines a hypersurface $V=\{ f=0 \}$ with $\dim(Y\cap V)\leq \dim(X)-2$,
and with $\xi$ and $\eta$ holomorphic forms.

\smallskip

In the following, we use the notation $^{S}\Omega^\star_X(\log(Y))$ to denote the
forms with logarithmic poles along $Y$ in the sense of Saito, to distinguish it from
the more restrictive notion of forms with logarithmic poles $\Omega^\star_X(\log(Y))$
considered above for the normal crossings case.

\smallskip

Following \cite{Ale2}, we say that a (reduced) hypersurface $Y\subset X$ has {\em Saito singularities}
if the modules of logarithmic differential forms and vector fields along $Y$ are free.
The condition that $Y\subset X$ has Saito singularities is equivalent to the condition
that $^{S}\Omega^n_X(\log(Y))= \bigwedge^n\,\, {}^{S}\Omega^1_X(\log(Y))$, \cite{Sa}.

\smallskip

Let $Y$ be a hypersurface with Saito singularities and
let $\cM_Y$ denote the sheaf of germs of meromorphic functions on $Y$. Then
setting
\begin{equation}\label{Resfxi}
{\rm Res}(\omega)=\frac{1}{f}\, \xi  \, |_Y
\end{equation}
defines the residue as a morphism of $\cO_X$-modules, for all $q\geq 1$,
\begin{equation}\label{Reshom}
 {\rm Res}: {}^{S}\Omega^q_X(\log(Y)) \to \cM_Y \otimes_{\cO_Y} \Omega^{q-1}_Y.
\end{equation}

\smallskip

Unlike the case of smooth hypersurfaces and normal crossings divisors, in the
case of more general hypersurfaces with Saito singularities, the Saito 
residue of forms with logarithmic poles is not a holomorphic form, but only
a {\em meromorphic} form on $Y$.

\smallskip

For $Y\subset X$ a reduced hypersurface that is quasihomogeneous with
Saito singularities, a refinement of \eqref{Reshom}, which we view as the exact sequence
$$ 0 \to \Omega_X^q \to {}^{S}\Omega^q_X(\log(Y)) \stackrel{{\rm Res}}{\to} \cM_Y \otimes_{\cO_Y} \Omega^{q-1}_Y, $$
is given in \cite{Ale2}, where the image of the Saito Poincar\'e residue is more
precisely identified as ${\rm Res}{}^S \Omega^q_X(\log (Y))\simeq \omega^{q-1}_Y$,
where $\omega_Y^\bullet$ denotes the module of regular meromorphic differential forms
in the sense of \cite{Barlet}, with $\omega_Y^\bullet \subset j_* j^* \Omega_Y^\bullet$, where 
$j:S\hookrightarrow Y$ is the
inclusion of the singular locus. 
Namely, it is shown in \cite{Ale2} that one has, for all $q\geq 2$, an exact sequence
of $\cO_X$-modules
\begin{equation}\label{resseq}
0 \to \Omega_X^q \to {}^S \Omega_X^q(\log(Y))
\stackrel{{\rm Res}}{\longrightarrow} \omega_Y^{q-1} \to 0.
\end{equation}

\smallskip

It is natural to ask whether the extraction of polar part from
forms with logarithmic poles that we considered here
for the case of smooth hypersurfaces and normal crossings divisors
extends to more general singular hypersurfaces using Saito's
formulation.

\smallskip

\begin{ques}\label{Saitoquestion}
Is there a Rota--Baxter operator $T$ expressed in terms of the
Saito residue, in the case of a singular hypersurfaces $Y\subset X$ with Saito singularities?
\end{ques}

\medskip

We describe here a
possible approach to this question. We introduce an analog of the Rota--Baxter operator
considered above, given by the extraction of the polar part. The ``polar part" operator,
in this more general case,  
does not map $\Omega_X^{{\rm even}}(\log(Y))$ to itself, but we show below that it
gives a well defined Rota-Baxter operator of weight $-1$ on the space of Saito
forms $^{S}\Omega_X^{{\rm even}}(\log(Y))$, and that this operator is a derivation.

\smallskip

\begin{lem}\label{SYloclem}
The set $S_Y:=\{f\,|\, {\rm dim}(\{f=0\}\cap Y)\leq {\rm dim}(X)-2\}$
is a multiplicative set. Localization of the Saito forms with logarithmic poles gives
$S_Y^{-1}\, {}^{S}\Omega_X(\log(Y))={}^{S}\Omega_X(\log(Y))$.
\end{lem}

\proof
We have $V_{12}=\{f_1f_2=0\}=\{f_1=0\}\cup\{f_2=0\}$ and $${\rm dim}(Y\cap V_{12})={\rm dim}((Y\cap\{f_1=0\})\cup (Y\cap\{f_2=0\}))\leq{\rm dim}(X)-2,$$
since ${\rm dim}(Y\cap\{f_i=0\})\leq{\rm dim}(X)-2$ for $i=1,2$. Thus, for any $f_1,f_2\in S_Y$, we have $f_1f_2\in S_Y$.
Moreover, we have $1\in S_Y$, hence $S_Y$ is a multiplicative set. The localization of ${}^{S}\Omega^\star_X(\log(Y))$
at $S_Y$ is just ${}^{S}\Omega^\star_X(\log(Y))$ itself: in fact, for
$\tilde{f}^{-1}\omega\in S_Y^{-1}\, {}^{S}\Omega^\star_X(\log(Y))$,
with $\tilde{f}\in S_Y$ and $\omega\in {}^{S}\Omega^\star_X(\log(Y))$, expressed as in \eqref{YsingRes},
we have
$$ f\tilde{f}(\tilde{f}^{-1}\omega)=f\omega=\frac{dh}{h}\wedge\xi+\eta,$$
where $f\tilde{f}\in S_Y$, hence $\tilde{f}^{-1}\omega\in {}^{S}\Omega_X(\log(Y))$.
 \endproof

\smallskip

Given a form $\omega\in{}^{S}\Omega^\star_X(\log(Y))$,
which we can write as in \eqref{YsingRes}, the residue \eqref{Resfxi} is the image under
the restriction map $S_Y^{-1}\Omega^\star_X \to S_Y^{-1}\Omega^\star_Y$ of the form $f^{-1} \xi
\in S_Y^{-1}\Omega^\star_X$. Moreover, we have an inclusion $\Omega^\star_X \hookrightarrow
 {}^{S}\Omega^\star_X(\log(Y))$, which induces a map of the localizations
 $S_Y^{-1} \Omega^\star_X \hookrightarrow S_Y^{-1}\, {}^{S}\Omega^\star_X(\log(Y))={}^{S}\Omega^\star_X(\log(Y))$.
We can then define a linear operator
$$ T: {}^{S}\Omega^\star_X(\log(Y)) \to {}^{S}\Omega^\star_X(\log(Y))\wedge S_Y^{-1}\, \Omega^\star_X \hookrightarrow
{}^{S}\Omega^\star_X(\log(Y))\wedge S_Y^{-1}\, {}^{S}\Omega^\star_X(\log(Y))={}^{S}\Omega^\star_X(\log(Y))$$
given by
\begin{equation}\label{SaitoT}
T(\omega)=\frac{dh}{h} \wedge \frac{\xi}{f}, \ \ \
\text{ for } \ \  f\, \omega = \frac{dh}{h} \wedge \xi + \eta .
\end{equation}

\smallskip

\begin{lem}\label{DerTSaito}
The operator $T$ of \eqref{SaitoT} is a Rota--Baxter operator of weight $-1$ on ${}^{S}\Omega^{\rm even}_X(\log(Y))$,
which is just given by a derivation, satisfying
the Leibnitz rule $T(\omega_1 \wedge \omega_2)= T(\omega_1) \wedge \omega_2 + \omega_1 \wedge T(\omega_2)$.
\end{lem}

\proof Let
\begin{equation*}
f_1\, \omega_1 = \frac{dh}{h} \wedge \xi_1 + \eta_1\quad f_2\, \omega_2 = \frac{dh}{h} \wedge \xi_2 + \eta_2.
\end{equation*}
Then
$$f_1\, f_2\, \omega_1 \wedge \omega_2 = (\frac{dh}{h} \wedge \xi_1 + \eta_1) \wedge (\frac{dh}{h} \wedge \xi_2 + \eta_2) = \frac{dh}{h} \wedge
(\xi_1 \wedge \eta_2 + (-1)^p\eta_1\wedge\xi_2)+ \eta_1\wedge\eta_2,$$
where $\eta_1\in\Omega^p(X)$.
By Lemma \ref{SYloclem}, we know that $f_1 f_2 \in S_Y$. We have
$$T(\omega_1\wedge\omega_2)=\frac{dh}{h}\wedge(\frac{\xi_1}{f_1}\wedge \frac{\eta_2}{f_2}+(-1)^p\frac{\eta_1}{f_1}\wedge \frac{\xi_2}{f_2}).$$
Since
$$T(\omega_1)=\frac{dh}{h}\wedge\frac{\xi_1}{f_1},\quad {\rm and}\quad T(\omega_2)=\frac{dh}{h}\wedge\frac{\xi_2}{f_2},$$
we obtain
$$T(\omega_1)\wedge T(\omega_2)=\frac{dh}{h}\wedge\frac{\xi_1}{f_1}\wedge\frac{dh}{h}\wedge\frac{\xi_2}{f_2}=0. $$
Moreover, we have
$$T(\omega_1)\wedge\omega_2= (\frac{dh}{h}\wedge\frac{\xi_1}{f_1})\wedge\frac{dh}{h}\wedge\frac{\xi_2}{f_2}+\frac{dh}{h}\wedge\frac{\xi_1}{f_1}\wedge\frac{\eta_2}{f_2}
=\frac{dh}{h}\wedge\frac{\xi_1}{f_1}\wedge\frac{\eta_2}{f_2},$$
with $$f_1f_2(T(\omega_1)\wedge\omega_2)=\frac{dh}{h}\wedge\xi_1\wedge\eta_2,$$
and similarly, $$\omega_1\wedge T(\omega_2)=(-1)^p\frac{dh}{h}\wedge\frac{\eta_1}{f_1}\wedge \frac{\xi_2}{f_2},$$
hence $T$ satisfies the Leibnitz rule. The operator $T$ also satisfies
$T(T(\omega_1)\wedge\omega_2)=T(\omega_1)\wedge\omega_2$, and $T(\omega_1\wedge T(\omega_2))=\omega_1\wedge T(\omega_2)$, hence the condition that $T$ is a derivation is equivalent to the condition that it is a
Rota-Baxter operator of weight $-1$.
\endproof

\medskip

Correspondingly, we have
$$ (1-T)\omega = \omega - \frac{dh}{h}\wedge \frac{\xi}{f} =\frac{\eta}{f}   \in S_Y^{-1} \, \Omega_X^{\rm even}. $$
Under the restriction map $S_Y^{-1} \, \Omega_X^{\rm even} \to S_Y^{-1} \, \Omega_Y^{\rm even}$ we obtain
a form $(1-T)(\omega)|_Y$. It follows that we can define a ``subtraction of divergences"
operation on $\phi: \cH \to {}^{S}\Omega_X^{{\rm even}}(\log(Y))$
by taking $\phi_+: \cH \to {}^{S}\Omega^{\rm even}_X(\log(Y))$ given by $\phi_+(a)=(1-T) \phi(a)|_Y$,
for $a\in \cH$, which maps
$\phi(a)=\omega$ to $(1-T)\omega|_Y = f^{-1} \eta|_Y$,
where $f\, \omega = \frac{dh}{h}\wedge \xi + \eta$.
While this has subtracted the logarithmic pole along $Y$, it has also created
a new pole along $V=\{ f=0 \}$.
Thus, it results again in a meromorphic form. If we consider the restriction to
$Y$ of $\phi_+(a)= f^{-1}\, \eta|_Y$,
we obtain a meromorphic form with first order poles along a subvariety $V\cap Y$, which is by hypothesis of codimension at least one in $Y$. Thus, we can conceive of a more complicated renormalization method that progressively subtracts poles on subvarieties of increasing
codimension, inside the polar locus of the previous pole subtraction, by iterating this procedure.
A more detailed account of this iterative procedure and of possible applications to the 
setting of renormalization will be discussed elsewhere.

\section{Compactifications of $\GL_n$ and momentum space Feynman integrals}\label{KGLsec}

In this section, we restrict our attention to the case of compactifications of
$\PGL_\ell$ and of $\GL_\ell$ and we use a formulation of the parametric
Feynman integrals of perturbative quantum field theory in terms of
(possibly divergent) integrals on a cycle in the complement of the
determinant hypersurface \cite{AluMa}, to obtain a new method of
regularization and renormalization. This gives rise to a
renormalized integral that is a period of a mixed Tate motive, under certain conditions
on the graph and on the intersection of the big cell of the compactification with
a divisor $\Sigma_{\ell,g}$. We show that 
a certain loss of information can occur with respect to the usual physical
Feynman integral.

\subsection{The determinant hypersurface}

In the following we use the notation $\hat\cD_\ell$ and $\cD_\ell$, respectively,
for the affine  and the  projective determinant hypersurfaces. Namely, we consider
in the affine space $\A^{\ell^2}$, identified with the space of all $\ell\times \ell$-matrices,
with coordinates $(x_{ij})_{i,j=1,\ldots,\ell}$, the hypersurface
$$ \hat\cD_\ell =\{ \det(X)=0\,|\, X=(x_{ij}) \} \subset \A^{\ell^2}. $$
Since $\det(X)=0$ is a homogeneous polynomial in the variables $(x_{ij})$,
we can also consider the projective hypersurface $\cD_\ell \subset \P^{\ell^2-1}$.

\smallskip

The complement $\A^{\ell^2}\smallsetminus \hat\cD_\ell$ is identified with the
space of invertible $\ell\times \ell$-matrices, namely with $\GL_\ell$.

\subsection{The Kausz compactification of $\GL_n$}

We recall here some basic facts about the Kausz compactification $K\GL_n$
of $\GL_n$, following \cite{Kausz} and the exposition in \S 12 of \cite{MaTha}.

\smallskip

We first recall the Vainsencher compactification \cite{Vain} of $\PGL_\ell$.
Let $X_0=\P^{\ell^2-1}$ be the projectivization of the vector space of
square $\ell\times \ell$-matrices. Let $Y_i$ be the locus of matrices of
rank $i$ and consider the iterated blowups $X_i= {\rm Bl}_{\bar Y_i}(X_{i-1})$,
with $\bar Y_i$ the closure of $Y_i$ in $X_{i-1}$. The $Y_i$ are $\PGL_i$-bundles 
over a product of Grassmannians.
It is shown in Theorem 1 and (2.4) of \cite{Vain} that the $X_i$ are smooth, and that $X_{\ell -1}$
is a wonderful compactification of $\PGL_\ell$, in the sense of \cite{DeCoPro}.
One denotes by $\overline{\PGL_\ell}$ the wonderful compactification of
$\PGL_\ell$ obtained in this way. We also refer the reader to \S 12 of \cite{MaTha} 
for a quick overview of the main properties of the Vainsencher compactification.

\smallskip

The Kausz compactification \cite{Kausz} of $\GL_\ell$ is similar. One regards $\A^{\ell^2}$ as the big cell
in $\cX_0=\P^{\ell^2}$. The iterated sequence of blowups is given in this case by setting
$\cX_i ={\rm Bl}_{\bar\cY_{i-1} \cup \bar\cH_i}(\cX_{i-1})$, where $\cY_i \subset \A^{\ell^2}$ are
the matrices of rank $i$ and $\cH_i$ are the matrices at infinity (that is, in
$\P^{\ell^2-1}=\P^{\ell^2}\smallsetminus \A^{\ell^2}$) of rank $i$. The Kausz compactification
is $K\GL_\ell = \cX_{\ell-1}$. It is shown in Corollary 4.2 of \cite{Kausz}
that the $\cX_i$ are smooth and in Corollary 4.2 and Theorem 9.1 of \cite{Kausz} that the blowup 
loci are disjoint unions of loci with the following structure: 
the closure $\bar\cY_{i-1}$ in $\cX_{i-1}$ is a $K\GL_{i-1}$-bundle over a product of
Grassmannians and the closure $\bar\cH_i$ in in $\cX_{i-1}$ is a $\overline{\PGL_i}$-bundle
over a product of Grassmannians. Theorem 9.1 of \cite{Kausz} also shows that these
compactifications have a moduli space interpretation. 
An overview of these properties and of the relation between the 
Vainsencher and the Kausz compactifications is given in \S 12 of \cite{MaTha}.

\smallskip

As observed in \cite{MaTha}, the Kausz compactification is then the closure
of $\GL_\ell$ inside the wonderful compactification of $\PGL_{\ell+1}$, see  also
\cite{Huru}, Chapter 3, \S 1.4. The compactification $K\GL_\ell$ is smooth and projective over
${\rm Spec}(\Z)$ (Corollary 4.2 \cite{Kausz}).

\smallskip

The other property of the Kausz compactification that we will be using in the
following is the fact that the complement of the dense open set $\GL_\ell$
inside the compactification $K\GL_\ell$ is a normal crossing divisor
(Corollary 4.2 \cite{Kausz}).

\subsection{The virtual motive of the Kausz compactification}\label{virtKsec}

We organize the computation of the motive of the Kausz compactification in
three subsections, respectively dealing with the virtual motive (Grothendieck class),
the numerical motive, and the Chow motive. The main reason for providing separate
arguments, instead of giving only the strongest result about the Chow motive, 
is a pedagogical illustration of the difference between these three levels of motivic structure,
where one can see in a very explicit case what is needed to improve from one 
level to the next, and what are the implications (conditional and unconditional).
We begin with the Grothendieck class, which is usually more familiar,
especially in the mathematical physics setting.

\smallskip

We use the description recalled above of the Kausz compactification,
together with the blowup formula, to check that the virtual motive (class in
the Grothendieck ring) of the Kausz compactification is Tate.

\begin{prop}\label{KauszTateGr}
Let $K_0(\cV)$ be the Grothendieck ring of varieties (defined over $\Q$ or
over $\Z$) and let $\Z[\bL]\subset K_0(\cV)$ be the Tate subring generated by
the Lefschetz motive $\bL=[\A^1]$. For all $\ell\geq 1$ the class $[K\GL_\ell]$
is in $\Z[\bL]$. Moreover, let $\cZ_\ell$ be the normal crossings divisor
$\cZ_\ell = K\GL_\ell \smallsetminus \GL_\ell$. Then all the unions and
intersections of components of $\cZ_\ell$ have Grothendieck classes in
$\Z[\bL]$.
\end{prop}

\proof We use the blowup formula for classes in the Grothendieck ring: if
$\tilde\cX={\rm Bl}_{\cY}(\cX)$, where $\cY$ is of codimension $m+1$ in $\cX$,
then the classes satisfy
\begin{equation}\label{blowupGr}
 [\tilde\cX]=[\cX] + \sum_{k=1}^m [\cY] \bL^k .
\end{equation}
The Kausz compactification is obtained as an iterated blowup, starting with
a projective space, whose class is in $\Z[\bL]$ and blowing up at each step
a smooth locus that is a bundle over a product of Grassmannians with fiber
either a $K\GL_i$ or a $\overline{\PGL}_i$ for some $i<\ell$. The Grothendieck
class of a bundle is the product of the class of the base and the class of the fiber.
Classes of Grassmannians (and products of Grassmannians) are in $\Z[\bL]$.
The classes of the wonderful compactifications $\overline{\PGL}_i$ of $\PGL_i$
are also in $\Z[\bL]$, since it is known that the motive of these wonderful
compactifications are mixed Tate (this follows, for instance, from the cell 
decomposition given in Proposition 4.4. of \cite{HabRad}). Thus, it
suffices to assume, inductively, that the classes $[K\GL_i]\in \Z[\bL]$ for all $i<\ell$,
and conclude via the blowup formula that $[K\GL_\ell] \in \Z[\bL]$.

Consider then the boundary divisor $\cZ_\ell = K\GL_\ell \smallsetminus \GL_\ell$.
The geometry of the normal crossings divisor $\cZ_\ell$ is
described explicitly in \S 4 of \cite{Kausz}. It has components
$Y_i$ and $Z_i$, for $0\leq i\leq \ell$, that correspond to the blowup loci described
above. The multiple intersections $\cap_{i\in I} Y_i \cap \cap_{j\in J} Z_j$ of
these components of $\cZ_\ell$ are described in turn in terms of bundles over
products of flag varieties with fibers that are lower dimensional compactifications
$K\GL_i$ and $\overline{\PGL}_i$ and products. Again, flag varieties have cell
decompositions, hence their Grothendieck classes are in $\Z[\bL]$ and the
rest of the argument proceeds as in the previous case. If arbitrary intersections
of the components of $\cZ_\ell$ have classes in $\Z[\bL]$ then arbitrary unions
and unions of intersections also do by inclusion-exclusion in $K_0(\cV)$.
\endproof

\smallskip
\subsection{The numerical motive of the Kausz compactification}\label{numKsec}

Knowing that the Grothendieck class $[K\GL_\ell]$ is in the Tate subring
$\Z[\bL]\subset K_0(\cV)$ implies that the motive is of Tate type in the category 
of pure motives with respect to the numerical equivalence. More precisely, we have
the following.

\begin{prop}\label{nummotive}
Let $h_{\rm num}(K\GL_\ell)$ denote the motive of the Kausz
compactification $K\GL_\ell$ in the category of pure motives
over $\Q$, with the numerical equivalence relation. Then
$h_{\rm num}(K\GL_\ell)$ is in the (tensor) subcategory generated by
the Tate object. The same is true for arbitrary unions and
intersections of the components of the boundary divisor $\cZ_\ell$
of the compactification.
\end{prop}

\proof The same argument used in Proposition \ref{KauszTateGr} can
be upgraded at the level of numerical motives. We replace the blowup
formula \eqref{blowupGr} for Grothendieck classes with the corresponding
formula for motives, which follows (already at the level of Chow motives)
from Manin's identity principle, \cite{Man}:
\begin{equation}\label{blowupmot}
h(\tilde X)=  h(X) \oplus \bigoplus_{r=1}^m h(Y) \otimes \bL^{\otimes r},
\end{equation}
with $\tilde X = {\rm Bl}_Y(X)$ the blowup of a smooth subvariety
$Y\subset X$ of codimension $m+1$ in a smooth projective variety $X$,
and where $\bL=h^2(\P^1)$ is the Lefschetz motive. Moreover, we use
the fact that, for numerical motives, the motive of a locally trivial
fibration $X\to S$ with fiber $Y$ is given by the product
\begin{equation}\label{fibrnum}
h_{\rm num}(X)= h_{\rm num}(Y) \otimes h_{\rm num}(S) ,
\end{equation}
see Exercise 13.2.2.2 of \cite{Andre}.
The decomposition \eqref{fibrnum} allows us to describe the
numerical motives of the blowup loci of the iterated blowup
construction of $K\GL_\ell$ as products of numerical motives
of Grassmannians and of lower dimensional compactifications
$K\GL_i$ and $\overline{\PGL}_i$. The motive of a Grassmannian
can be computed explicitly as in \cite{Ko}, already at the level
of Chow motives. If $G(d,n)$ denotes the Grassmannian of
$d$-planes in $k^n$, the Chow motive $h(G(d,n))$ is given by
\begin{equation}\label{MotGrass}
h(G(d,n))=\bigoplus_{\lambda \in W^d} \bL^{\otimes |\lambda|},
\end{equation}
where
$$ W^d =\{ \lambda=(\lambda_1,\ldots, \lambda_d) \in \N^d\,|\, n-d \geq \lambda_1 \geq
\cdots \geq \lambda_d \geq 0 \} $$
and $|\lambda|=\sum_i \lambda_i$, see Theorem 2.1 and Lemma 3.1 of \cite{Ko}.
The same decomposition into powers of the Lefschetz motive holds at the
numerical level. Moreover, we know (also already for Chow motives) that
the motives $h(\overline{\PGL}_i)$ of the wonderful compactifications
are Tate (see  \cite{HabRad}), and we conclude the argument as in
Proposition \ref{KauszTateGr} by assuming inductively that the
motives $h_{\rm num}(K\GL_i)$ are Tate, for $i<\ell$. The argument for
the loci $\cap_{i\in I} Y_i \cap \cap_{j\in J} Z_j$ in $\cZ_\ell$ is analogous.
\endproof

\begin{rem}\label{remNumK0}{\rm
Proposition \ref{nummotive} also follows from Proposition \ref{KauszTateGr} using
the general fact that two numerical motives that have the same class in
$K_0({\rm Num}(k)_\Q)$ are isomorphic as objects in ${\rm Num}(k)_\Q$,
because of the semi-simplicity of the category of numerical motives \cite{Jan}, together
with the existence, for $char(k)=0$, of a unique
ring homomorphism (the motivic Euler characteristic) $\chi_{\rm mot}: K_0(\cV_k)\to
K_0({\rm Num}(k)_\Q)$, such that $\chi_{\rm mot}([X])=[h_{\rm num}(X)]$, for $X$ 
a smooth projective variety, see Corollary 13.2.2.1 of \cite{Andre}. }
\end{rem}

\smallskip
\subsection{The Chow motive of the Kausz compactification}\label{chowKsec}

Manin's blowup formula \eqref{blowupmot} and the computation of the motive
of Grassmannians and of the wonderful compactifications $\overline{\PGL}_i$
already hold at the level of Chow motives. However, if we want to extend the 
argument of Proposition \ref{nummotive} to Chow motives, we run into the 
additional difficulty that one no longer necessarily has the
decomposition \eqref{fibrnum} for the motive of a locally trivial fibration. 
Under some hypotheses on the existence of a cellular structure on the fibers, one can still
obtain a decomposition for motives of bundles, and more generally locally
trivial fibrations, the fibers of which have cell decompositions with suitable properties,
see \cite{Karp}, and also \cite{Habi}, \cite{HabRad}, \cite{Ka}, \cite{PoGu}. 
We obtain an unconditional result on the Chow motive of the Kausz
compactification, by analyzing its cellular structure.

\smallskip

Recall that, for $G$ a connected reductive algebraic group and $B$ a Borel subgroup, 
a {\em spherical variety} is a normal algebraic variety on which $G$ acts with a 
dense orbit of $B$, \cite{Brion}. Spherical varieties can be regarded as a generalization
of toric varieties: when $G$ is a torus, one recovers the usual notion of toric variety.

\begin{prop}\label{Chowprop}
The Chow motive $h(K\GL_\ell)$ of the Kausz compactification is a Tate motive.
\end{prop}

\proof The result follows by showing that $K\GL_\ell$ has a cellular structure for all $\ell\geq 1$,
which allows us to extend the decomposition of the motive used in Proposition \ref{nummotive}
from the numerical to the Chow case. 

As shown in \S 3.1 of \cite{Brion}, it follows from the work of Bialynicki--Birula 
\cite{BiaBir} that any complete, smooth and spherical variety $X$ has a cellular 
decomposition. This is determined by the decomposition of the spherical
variety into $B$-orbits and is obtained by considering a one-parameter
subgroup $\lambda: \bG_m \hookrightarrow X$, such that the set of fixed points $X^\lambda$
is finite. The cells are given by
\begin{equation}\label{Xlambdax}
 X(\lambda,x) =\{ z\in X\,|\, \lim_{t\to 0} \lambda(t) z =x \}, \ \ \text{for} \ \ x\in X^\lambda. 
\end{equation} 

The Kausz compactification $K\GL_\ell$ is a smooth toroidal 
equivariant compactification of $\GL_\ell$, see Proposition 1.15 of \S 3 
of \cite{Huru} and also Proposition 10.1 and Proposition 12.1 of \cite{MaTha}.
In particular, it is a spherical variety (see Proposition 10.1 of \cite{MaTha}), 
hence it has a cellular structure as above.

\smallskip

A relative cellular variety, in the sense of \cite{Karp}, is a smooth
and proper variety with a decomposition into affine fibrations over
proper varieties. The blowup loci of the Kausz compactification $K\GL_\ell$
are relative cellular varieties in this sense, since they are bundles
over products of Grassmannians, with fiber a lower dimensional
compactification $K\GL_i$, with $i<\ell$. Using the cell decomposition
of the fibers $K\GL_i$, we obtain a decomposition of these blowup
loci as relative cellular varieties, with pieces of the decomposition
being fibrations over a product of Grassmannians, with fibers the
cells of the cellular structure of $K\GL_i$. 

\smallskip

There is an embedding of the category of pure Chow motives in
the category of mixed motives, see \cite{Andre}. By viewing the 
Chow motives of these blowup loci 
as elements in the Voevodsky category of mixed motives, 
Corollary 6.11 of \cite{Karp} shows that they are direct sums of motives of 
products of Grassmannians (which are Tate motives), 
with twists and shifts which depend on the dimensions of the cells
of $K\GL_i$. We conclude from this that all the blowup loci
are Tate motives. We can then repeatedly apply the blowup formula
for Chow motives to conclude (unconditionally) that the Chow motive of $K\GL_\ell$ 
is itself a Tate motive. Note that the blowup formula also 
holds in the Voevodsky category, Proposition 3.5.3 of \cite{Voev}, in the form
$$ \m({\rm Bl}_Y(X)) = \m(X)\oplus \bigoplus_{r=1}^{{\rm codim}_X(Y)-1} \m(Y)(r)[2r], $$
which corresponds to the usual formula of \cite{Man} in the case of pure motives, after
viewing them as objects in the category of mixed motives. The result can also be obtained, in a similar way, 
using Theorem 2.10 of \cite{HabRad} instead of Corollary 6.11 of \cite{Karp}.
\endproof

\begin{rem}\label{cellChowrem}{\rm
Given the existence of a cellular decomposition of $K\GL_\ell$, as above,
it is possible to give a quicker proof that the Chow motive is Tate, by
using distinguished triangles in the Voevodsky category associated to
the inclusions of unions of cells, showing that $\m(K\GL_\ell)$ is
mixed Tate, then using the inclusion of pure motives in the mixed
motives to conclude that $h(K\GL_\ell)$ is Tate. In Proposition \ref{Chowprop}
above we chose to maintain the structure of the argument more similar to the cases
of the virtual and the numerical motive, for better direct comparison.}
\end{rem}

\begin{rem}\label{remKS}{\rm
Notice that a {\em conditional} result about the Chow motive would follow
directly from Proposition \ref{nummotive} or Remark \ref{remNumK0}, if one assumes the
Kimura--O'Sullivan finiteness conjecture (or Voevodsky's nilpotence conjecture, which implies it).
For the precise statement and implications of the Kimura--O'Sullivan finiteness conjecture, and its
relation to Voevodsky's nilpotence, we refer 
the reader to the survey \cite{Andre2}.
By arguing as in Lemma 13.2.1.1 of \cite{Andre}, that would extend 
the result of Proposition \ref{nummotive} to the Chow motive. At the level of
Grothendieck classes, 
the conjecture in fact implies that the $K_0$ of Chow motives and the $K_0$ of
numerical motives coincide, hence one can argue as in Remark \ref{remNumK0}
and conclude that, in order to know that the Chow motive is mixed Tate,
it suffices to know that the Grothendieck class is mixed Tate.}
\end{rem}

\subsection{Feynman integrals in momentum space and non-mixed-Tate examples}

It was shown in \cite{BEK} that the parametric form of Feynman integrals in
perturbative quantum field theory can be formulated as a (possibly divergent)
period integral on the complement of a hypersurface defined by the vanishing
of a combinatorial polynomial associated to Feynman graphs. Namely,
one writes the (unrenormalized) Feynman amplitudes for a {\em massless}
scalar quantum field theory as integrals
\begin{equation}\label{paramInt}
 U(\Gamma)=   \frac{\Gamma(n-D\ell/2)}{(4\pi)^{\ell D/2}}
\int_{\sigma_n} \frac{P_\Gamma(t,p)^{-n +D\ell/2} \omega_n}
{\Psi_\Gamma(t)^{-n + D(\ell +1)/2}}
\end{equation}
where $n=\# E_\Gamma$ is the number of internal edges, $\ell=b_1(\Gamma)$
is the number of loops, and $D$ is the spacetime dimension. Here we consider
the ``unregularized" Feynman integral, where $D$ is just the integer valued dimension,
without performing the procedure of dimensional regularization that analytically
continues $D$ to a complex number. 
The domain of integration is a simplex
$\sigma_n=\{ t\in \R^n_+ | \sum_i t_i=1\}$. In the integrand,
$\omega_n$ is the volume form, and $P_\Gamma$ and $\Psi_\Gamma$
are polynomials defined as follows.
The graph polynomial is defined as
$$ \Psi_\Gamma(t)=\sum_T \prod_{e\notin T} t_e $$
where the summation is over spanning trees (assuming the graph $\Gamma$ is connected).
The polynomial $P_\Gamma$ is given by
\begin{equation}\label{PGammapt}
 P_\Gamma(t,p) = \sum_{C\subset \Gamma} s_C \prod_{e\in C} t_e 
\end{equation} 
with the sum over cut-sets $C$ (complements of a spanning tree plus one edge)
and with variables $s_C$ depending on the external momenta of the graph,
$s_C =(\sum_{v\in V(\Gamma_1)} P_v)^2$, where $\Gamma_1$ is one of
the connected components after the cut (by momentum conservation, it does not matter which).
The variables $P_v$ are given by combinations of the external momenta 
$p=(p_e)\in \Q^{\# E_{ext}(\Gamma)\cdot D}$, in the form 
$P_v=\sum_{e\in E_{ext}(\Gamma), t(e)=v} p_e$, where
$\sum_{e\in E_{ext}(\Gamma)} p_e =0$. 

\smallskip

In the range $-n+D\ell/2 \geq 0$,
which includes the log divergent case $n=D\ell/2$, the Feynman
amplitude is therefore the integral of an algebraic differential form
defined on the complement of the graph hypersurface
$\hat X_\Gamma =\{ t\in \A^n \,|\, \Psi_\Gamma(t)=0 \}$.
Divergences occur due to the intersections of the domain of integration $\sigma_n$
with the hypersurface. Some regularization and renormalization procedure is
required to separate the chain of integration from the divergence locus. We refer the
reader to \cite{BEK} (or to \cite{Mar} for an introductory exposition).

\smallskip

It was originally conjectured by Kontsevich that the graph hypersurfaces $\hat X_\Gamma$
would always be mixed Tate motives, which would have explained the pervasive occurrence
of multiple zeta values in Feynman integral computations observed in \cite{BrKr}. A general
result of \cite{BeBro} disproved the conjecture, while more recent results of \cite{BrDo},
\cite{BrSch}, \cite{Dor} showed explicit examples of Feynman graphs that give rise to
non-mixed-Tate periods.

\subsection{Determinant hypersurface and parametric Feynman integrals}\label{DetFeynSec}

In \cite{AluMa} the computation of parametric Feynman integrals was
reformulated by replacing the graph hypersurface complement by
the complement of the determinant hypersurface.

\smallskip

More precisely, the (affine) graph hypersurface $\hat X_\Gamma$ is defined by the vanishing
of the graph polynomial $\Psi_\Gamma$. It follows from the matrix-tree theorem that this
polynomial can be written as a determinant
$$ \Psi_\Gamma(t)=\det M_\Gamma(t)=\sum_T \prod_{e\notin T} t_e \, , $$
with $M_\Gamma(t)$ the $\ell\times \ell$ matrix
\begin{equation}\label{MGamma}
 (M_\Gamma)_{kr}(t)=\sum_{i=1}^n t_i \eta_{ik} \eta_{ir}, 
\end{equation} 
where the matrix $\eta$ is given by
$$ \eta_{ik}=\left\{ \begin{array}{ll} \pm 1 & \text{edge } \pm  e_i \in \text{ loop } \ell_k \\[2mm]
0 & \text{otherwise.} \end{array} \right. $$
This definition of the matrix $\eta$ involves the choice of a basis $\{ \ell_k \}$ of the
first homology $H_1(\Gamma;\Z)$ and the choice of an orientation of the edges of the graph,
with $\pm e$ denoting the matching/reverse orientation on the edge $e$. The resulting
determinant $\Psi_\Gamma(t)$ is independent of both choices. 

\smallskip

One considers then the map
$$ \Upsilon: \A^n \to \A^{\ell^2}, \ \ \
\Upsilon(t)_{kr}=\sum_i t_i \eta_{ik} \eta_{ir} $$
that realizes the graph hypersurface as the preimage
$$ \hat X_\Gamma = \Upsilon^{-1}(\hat \cD_\ell) $$
of the determinant hypersurface $\hat\cD_\ell =\{ \det(x_{ij})=0 \}$.

It is shown in \cite{AluMa} that the map
\begin{equation}\label{Upsilonmap}
 \Upsilon: \A^n \smallsetminus \hat X_\Gamma \ \hookrightarrow \ \A^{\ell^2}
\smallsetminus \hat\cD_\ell
\end{equation}
is an embedding whenever the graph $\Gamma$ is 3-edge-connected with a
closed 2-cell embedding of face width $\geq 3$.

\smallskip

\begin{rem}{\rm
As discussed in \S 3 of \cite{AluMa}, the 3-edge-connected condition on
graphs can be viewed as a strengthening of the usual 1PI (one-particle-irreducible)
condition assumed in physics, since the 1PI condition corresponds to 2-edge-connectivity.
In perturbative quantum field theory, one considers 1PI graphs when
computing the asymptotic expansion of the effective action. Similarly, one
can consider the 2PI effective action (which is related to non-equilibrium
phenomena in quantum field theory, see \S 10.5.1 of \cite{Rammer})
and restrict to 3-edge-connected graphs. The condition of having a
closed 2-cell embedding of face width $\geq 3$, on the other hand, is
a strengthening of the analogous property with face width $\geq 2$,
which conjecturally is satisfied for all 2-vertex-connected graphs
(strong orientable embedding conjecture, see 
Conjecture 5.5.16 of \cite{MoTho}). 
2-vertex-connectivity is again a natural strengthening of the 1PI condition.
A detailed discussion of equivalent formulations and implications of 
these combinatorial conditions, as well as specific examples of graphs 
that fail to satisfy them, are given in \S 3 of \cite{AluMa}. }
\end{rem}

\smallskip

Let $\cP_\Gamma(x,p)$ denote a homogeneous polynomial in $x\in \A^{\ell^2}$, with 
$p\in \Q^{\# E_{ext}(\Gamma)\cdot D}$,
with the property that the restriction to the image of the map $\Upsilon=\Upsilon_\Gamma$
agrees with the second Symanzik polynomial $P_\Gamma$ defined in \eqref{PGammapt},
$$\cP_\Gamma(x,p)|_{x=\Upsilon(t)\in \Upsilon(\A^n)} = P_\Gamma(t,p). $$  
When the map $\Upsilon_\Gamma$ is an embedding, one can, without loss of information,
rewrite the parametric Feynman integral as (see Lemma 2.3 of \cite{AluMa})
\begin{equation}\label{UintDet}
 U(\Gamma) = \int_{\Upsilon(\sigma_n)} \frac{\cP_\Gamma(x,p)^{-n+D\ell/2} \omega_\Gamma(x)}
{\det(x)^{-n+(\ell+1)D/2}}.
\end{equation}
Here $\omega_\Gamma(x)$ is an $n$-form on $\A^{\ell^2}$ such that 
the restriction of $\omega_\Gamma(x)$ to the subspace $\Upsilon(\A^n)$ 
satisfies $\omega_\Gamma(\Upsilon(t)) =\omega_n(t)$, 
with $\omega_n$ the volume form on $\A^n$. 
Under the condition that $\Upsilon$ is an embedding, the restriction of the integrand 
to the image $\Upsilon(\sigma_n)$ then agrees with the original Feynman integral. 

\smallskip

The question on the nature of periods is then reformulated in \cite{AluMa} by
considering a normal crossings divisor $\hat\Sigma_\Gamma$ in $\A^{\ell^2}$ with
$\Upsilon(\partial\sigma_n)\subset \hat\Sigma_\Gamma$ and considering the
motive
\begin{equation}\label{motiveDetDiv}
 \m (\A^{\ell^2}\smallsetminus \hat\cD_\ell, \hat\Sigma_\Gamma \smallsetminus
(\hat\Sigma_\Gamma \cap \hat\cD_\ell)).
\end{equation}

The motive $\m(\A^{\ell^2} \smallsetminus \hat\cD_\ell)$
of the determinant hypersurface complement 
belongs to the category of mixed Tate motives (see Theorem 4.1 of \cite{AluMa}), 
with Grothendieck class
$$ [\A^{\ell^2} \smallsetminus \hat\cD_\ell]=\bL^{\binom{\ell}{2}} \prod_{i=1}^\ell (\bL^i-1). $$
However, as shown in \cite{AluMa}, the nature of the motive \eqref{motiveDetDiv} is much
more difficult to discern, because of the nature of the intersection between the divisor
$\hat\Sigma_\Gamma$ and the determinant hypersurface. Assuming the previous 
conditions on the graph (3-edge-connectedness with a closed 2-cell embedding of face 
width at least $3$), it is shown in Proposition 5.1 of
\cite{AluMa} that one can consider a divisor $\hat\Sigma_{\ell,g}$ that only depends on
$\ell =b_1(\Gamma)$ and on the minimal genus $g$ of the surface $S_g$
realizing the closed 2-cell embedding of $\Gamma$,
\begin{equation}\label{Sigmaellg}
 \hat\Sigma_{\ell,g} = L_1 \cup \cdots \cup L_{\binom{f}{2}},
\end{equation}
where $f=\ell-2g+1$ and the irreducible components $L_1, \ldots, L_{\binom{f}{2}}$
are linear subspaces defined by the equations
$$ \left\{ \begin{array}{rll} x_{ij}& =0 & 1\leq i < j \leq f-1 \\[2mm]
x_{i1}+ \cdots + x_{i,f-1} &  =0 & 1 \leq i \leq f-1. \end{array} \right. $$

For a given choice of subspaces $V_1, \ldots, V_\ell$ of a fixed $\ell$-dimensional space, 
one defines the {\em variety of frames} as 
$$ \bF(V_1,\ldots,V_\ell) := \{ (v_1,\ldots,v_\ell)\in \A^{\ell^2}
\smallsetminus \hat\cD_\ell \,|\, v_k \in V_k \}.   $$
In other words, the variety of frames $\bF(V_1,\ldots,V_\ell)$ consists of the
set of $\ell$-tuples $(v_1, \ldots, v_\ell)$, with $v_i\in V_i$, such that 
$\dim {\rm span}\{ v_1, \ldots, v_\ell\}=\ell$.  
It is shown in \cite{AluMa} that the motives
\eqref{motiveDetDiv} are mixed Tate if the varieties of frames
$\bF(V_1,\ldots,V_\ell)$ are mixed Tate. This question is closely 
related to the geometry of intersections of unions of
Schubert cells in flag varieties and Kazhdan--Lusztig theory.

\smallskip

In this paper we will follow a different approach, which uses the same reformulation
of parametric Feynman integrals in momentum space in terms of determinant
hypersurfaces, as in \cite{AluMa}, but instead of computing the integral in the
determinant hypersurface complement, pulls it back to the Kausz compactification of $\GL_\ell$,
following the model of computations of Feynman integrals in configuration spaces 
described in \cite{CeMa}.

\medskip
\subsection{Cohomology and forms with logarithmic poles}

Let $\cX$ be a smooth projective variety and $\cZ \subset \cX$ a divisor. 
Let $\cM_{\cX,\cZ}^\star$ denote, as before, the
complex of meromorphic differential forms on $\cX$ with poles (of arbitrary 
order) along $\cZ$, and let $\Omega_\cX^\star(\log(\cZ))$ be the complex
of forms with logarithmic poles along $\cZ$. Let $\cU=\cX\smallsetminus \cZ$ 
and $j: \cU\hookrightarrow \cX$ be the inclusion. 

\smallskip

Grothendieck's Comparison Theorem, \cite{Groth}, shows that the natural
morphism (de Rham morphism) 
$$ \cM_{\cX,\cZ}^\star \to R j_* \C_{\cU} $$
is a quasi-isomorphism, hence de Rham cohomology $H^\star_{dR}(\cU)$
is computed by the hypercohomology of the meromorphic de Rham complex. 
In particular, for $\cU$ affine, the hypercohomology
is not necessary and all classes are represented by closed global differential
forms, with hypercohomology replaced by the cohomology of the complex
of global sections. 

\smallskip

The {\em Logarithmic Comparison Theorem} consists of the statement
that, for certain classes of divisors $\cZ$, the natural morphism
$$ \Omega_\cX^\star(\log(\cZ)) \to \cM_{\cX,\cZ}^\star $$
is also a quasi-isomorphism. This is known to hold for simple
normal crossings divisors by \cite{Del}, and for strongly quasihomogeneous
free divisors by \cite{CJNMM}, and for a larger class of locally quasihomogeneous
divisors in \cite{HoMo}. For our purposes, we will focus only on the
case of simple normal crossings divisors.

\smallskip

In combination with Grothendieck's Comparison Theorem, the Logarithmic
Comparison Theorem of \cite{Del} for a simple normal crossings divisor 
implies that the de Rham cohomology of the divisor
complement is computed by the hypercohomology of the logarithmic de Rham
complex,
\begin{equation}\label{loghypH}
 H^\star_{dR}({\mathcal U}) \simeq {\mathbb H}^\star({\mathcal X},
\Omega_{{\mathcal X}}^\star(\log {\mathcal Z})). 
\end{equation}

\smallskip

\begin{rem}\label{hyponly}{\rm Even under the assumption that
the complement $\cU$ is affine, the hypercohomology on the
right hand side of \eqref{loghypH} cannot always be replaced
by global sections and cohomology. For example, if $\cX$ is a
smooth projective curve of genus $g$, and $\cU$ is the 
complement of $n$ points in $\cX$, then $H^1_{dR}(\cU)$
has dimension $2g+n-1$, but the dimension of the space of
global sections of the sheaf of logarithmic differentials is only
$g+n-1$ by Riemann-Roch.}
\end{rem}

Some direct comparisons between de Rham cohomology $H^\star_{dR}(\cU)$
and the cohomology of the logarithmic de Rham complex are known. We
discuss in the coming subsections how these apply to our specific case.
Our purpose is to replace the meromorphic form that arises in the
Feynman integral computation with a cohomologous form with logarithmic
poles along the divisor of the Kausz compactification. In doing so,
we need to maintain explicit control of the motive of the variety over
which cohomology is taken, and also maintain the algebraic nature
of all the differential forms involved. 

\subsection{Pullback to the Kausz compactification, forms with
logarithmic poles, and renormalization}

For fixed $D, \ell \in \N$ (respectively the integer spacetime dimension and the loop number) 
and for assigned external momenta $p \in \Q^D$, we
now consider the algebraic differential form
\begin{equation}\label{eta}
\eta_{\Gamma,D,\ell,p}(x):=\frac{\cP_\Gamma(x,p)^{-n+D\ell/2} \omega_\Gamma(x)}
{\det(x)^{-n+(\ell+1)D/2}}.
\end{equation}
For simplicity, we write the above as $\eta_\Gamma(x)$. This is defined
on the complement of the determinant hypersurface,
$\A^{\ell^2}\smallsetminus \hat\cD_\ell =\GL_\ell$.  Thus, by pulling back to the
Kausz compactification, we can regard it as an algebraic differential form on
$$ K\GL_\ell \smallsetminus \cZ_\ell = \GL_\ell, $$
where $\cZ_\ell$ is the normal crossings divisor at the boundary of the
Kausz compactification.

\smallskip
\subsubsection{Cellular decomposition approach}\label{cellrenSec}

We consider a special case of a simple normal crossings divisor $\cZ$
in a smooth projective variety $\cX$, under the additional assumption that 
$\cX$ has a cell decomposition. We denote by $\{ X_{\alpha, i} \}$ the
finite collection of cells of dimension $i$, and in particular we simply
write $X_\alpha =X_{\alpha, \dim\cX}$ for the top dimensional cells.

\begin{prop}\label{periodcells}
Let $\cZ\subset \cX$ be a pair as above, with $\{ X_\alpha \}$ the top dimensional 
cells of the cellular decomposition. Given a meromorphic form 
$\eta \in \cM^m_{\cX,\cZ}$, there exist forms 
$\beta^{(\alpha)}$ on $X_\alpha$ with logarithmic
poles along the normal crossings divisor $\cZ$, such that
\begin{equation}\label{cohomcell}
[\beta^{(\alpha)}]=[\eta |_{X_\alpha}] \in H^*_{dR} (X_\alpha \smallsetminus \cZ).
\end{equation}
\end{prop}

\proof Lemma 2.5 of \cite{CJNMM} shows that the Logarithmic Comparison Theorem is
equivalent to the statement that, for all Stein open sets $\cV\subset \cX$,
there are isomorphisms $H^\star(\Gamma(\cV,\Omega_\cX^\star(\log\cZ))) \simeq
H^\star_{dR}(\cV\smallsetminus \cZ)$. Namely, the hypercohomology in the
Logarithmic Comparison Theorem can be replaced by cohomology of the complex
of sections, when restricted to Stein open sets. 
\endproof

\begin{rem}\label{betaalphaCech}{\rm
The forms $\beta^{(\alpha)}$ do not match consistently on the
closures of the cells $X_\alpha$, because of nontrivial \v{C}ech cocycles, hence they are not
restrictions of a unique form with logarithmic poles $\beta$ defined on
all of $\cX$. In particular, the forms $\beta^{(\alpha)}$ obtained in this way depend on
the cellular decomposition used. 
}\end{rem}

\smallskip  

\begin{lem}\label{intXalpha}
Let $\cZ\subset \cX$ and $\{ X_\alpha \}$ be as above, and suppose given
a meromorphic form $\eta \in \cM^N_{\cX,\cZ}$, with $N=\dim \cX$, and an $N$-chain
$\sigma \subset \cX$ with $\partial\sigma\subset \Sigma$, for a divisor $\Sigma$ in $\cX$. 
After performing a pole subtraction on the logaritmic
forms on each cell $X_\alpha$ one can replace the integral $\int_\Sigma \eta$
with a renormalized version
\begin{equation}\label{renXalint}
\int_\sigma \beta^+ := \sum_\alpha  \int_{X_\alpha \cap \sigma} \beta^{(\alpha),+},
\end{equation}
where $\beta^{(\alpha),+}$ is a simple pole subtraction on $\beta^{(\alpha)}$.
The integral \eqref{renXalint} is a sum of periods of motives $\m(X_\alpha, X_\alpha\cap \Sigma)$.
The information contained in the subtracted polar part is recovered by the Poincar\'e
residues 
\begin{equation}\label{ResZal}
\int_{\sigma\cap \cZ_I} {\rm Res}_{\cZ_I}(\beta):= \sum_\alpha \int_{\sigma\cap \cZ_I \cap X_\alpha} {\rm Res}_{\cZ_I}(\beta^{(\alpha)}) 
\end{equation} 
along the intersections of components $\cZ_I= Z_{i_1}\cap \cdots \cap Z_{i_k}$, $I=\{ i_1,\ldots, i_k\}$ 
of the divisor $\cZ$. These are sums of periods of the motives $\m(\cZ_I \cap X_\alpha)$.
\end{lem}

\proof Given the cell decomposition as above, we can write the integral as
\begin{equation}\label{intoncells1}
\int_{\sigma}  \eta = \sum_\alpha \int_{X_\alpha \cap \sigma}  \eta |_{X_\alpha} 
= \sum_\alpha \int_{X_\alpha \cap \sigma} \beta^{(\alpha)} ,
\end{equation}
where each $\eta |_{X_\alpha}$ is replaced by the cohomologous $\beta^{(\alpha)}$
with logarithmic poles. After performing a pole subtraction on each $\beta^{(\alpha)}$
we obtain holomorphic forms $\beta^{(\alpha),+}$, hence the resulting integral is
a period of  $\m(X_\alpha, X_\alpha\cap \Sigma)$. For the relation between polar
subtraction and the Poincar\'e residues, see the discussion in \S \ref{ResSec} 
and \S \ref{multiresSec} above.
\endproof

In both \eqref{renXalint} and \eqref{ResZal}, we use the notation on the left-hand-side, with
a global integral and a global form $\beta$,  purely as a formal shorthand notation for 
the sum of the integrals on the cells of the $\beta^{(\alpha)}$,
since the latter are not restrictions of a global form $\beta$.

\begin{rem}\label{periodnosum}{\rm Notice that the resulting integral \eqref{renXalint}
obtained in this way can be identified with a period of $\m(\cX,\Sigma)$ only in the case
where the forms $\beta^{(\alpha),+}$ are restrictions $\beta^{(\alpha),+}=\beta^+ |_{X_\alpha}$
of a single holomorphic form $\beta^+$ on $\cX$. More generally, the resulting \eqref{renXalint}
is only a sum of periods of the motives $\m(X_\alpha, X_\alpha\cap \Sigma)$.}\end{rem}

\begin{rem}\label{onecell}{\rm If the cellular decomposition of $\cX$ has a single top
dimensional cell $X$, then a unique form with logarithmic poles 
$\beta \in \Omega_X^\star(\log \cZ)$, satisfying $[\eta|_X]=[\beta]\in H^\star_{dR}(X\smallsetminus \cZ)$, 
suffices to regularize the integral $\int_\sigma \eta$, with regularized value $\int_{\sigma\cap X} \beta^+$.
}\end{rem}

\smallskip

As we discussed in Proposition \ref{Chowprop}, the Kausz compactification
is a spherical variety (Proposition 1.15 of \S 3 
of \cite{Huru} and also Propositions 9.1, 10.1 and 12.1 of \cite{MaTha}),
hence it has a cellular decomposition (\S 3.1 of \cite{Brion}) into cells $X(\lambda,x)$
as in \eqref{Xlambdax}. Thus, we can apply the procedure described above,
to regularize the integral $\int_{\Upsilon(\sigma)} \eta_\Gamma$. 
While this regularization procedure depends on the choice of the cell decomposition,
the construction of \cite{Brion} for spherical varieties provides a cellular structure
that is intrinsically defined by the orbit structure of $K\GL_\ell$ 
and is quite naturally reflecting its geometry.
We can then perform a renormalization procedure based on the pole subtraction
procedure for forms with logarithmic poles described above.

\begin{cor}\label{intcellslog}
The cell decomposition $\{ X(\lambda,x) \}$ of $K\GL_\ell$ has a single big cell $X$.
Given $\eta_\Gamma = \eta_{\Gamma,D,\ell,p}$ as in \eqref{eta}, there is a form $\beta_\Gamma =
\beta_{\Gamma,D,\ell,p}$ on the big cell $X$, with logarithmic poles along $\cZ_\ell$,
such that $[\eta_\Gamma |_X ]=[\beta_\Gamma]\in H^\star_{dR}(X\smallsetminus \cZ)$.
Applying the Birkhoff factorization for forms with logarithmic poles to $\beta_\Gamma$, we
obtain a renormalized integral of the form
\begin{equation}\label{renormXint}
R(\Gamma) = \int_{\tilde\Upsilon(\sigma_n)\cap X} \beta^+_{\Gamma,D,\ell,p},
\end{equation}
where $\beta^+_\Gamma$ is a simple pole subtraction on $\beta_\Gamma$.
\end{cor}

\proof As mentioned in Proposition \ref{Chowprop}, the spherical variety
$K\GL_\ell$ is a smooth toroidal equivariant compactification of $\GL_\ell$
(Proposition 1.15 of \S 3 of \cite{Huru} and Propositions 9.1 and 12.1 of \cite{MaTha}).
By \S 2.3 of \cite{BrionLuna} and Proposition 9.1 of \cite{MaTha}, it then follows
that there is just one big cell $X$. We can then write the integral as
\begin{equation}\label{intoncells}
\int_{\tilde\Upsilon(\sigma_n)}  \eta_\Gamma = \int_{X \cap \tilde\Upsilon(\sigma_n)}  \eta_\Gamma |_X ,
\end{equation}
where $\tilde\Upsilon(\sigma_n)$ is the pullback to $K\GL_\ell$ of the
domain of integration $\Upsilon(\sigma_n)$.

\smallskip

Let $\cH$ be the Hopf algebra of Feynman graphs.
The morphism $\phi: \cH \to \cM^*_{X, \cZ_\ell \cap X}$
assigns to a Feynman graph $\Gamma$ a meromorphic
differential form $\beta_\Gamma=\beta_{\Gamma,D,\ell,p}$
with logarithmic poles along $\cZ_\ell$ satisfying
$[\eta_\Gamma |_X ]=[\beta_\Gamma]\in H^\star_{dR}(X\smallsetminus \cZ)$.

\smallskip

We then perform the Birkhoff factorization, and we denote by
$\beta^+_\Gamma$ the regular differential form on $X \subset K\GL_\ell$
given by $\phi^+(\Gamma)=\beta^+_\Gamma$.
Since we only have logarithmic poles, by Proposition \ref{Tlogid}
the operation becomes a simple pole subtraction and we have
$\beta^+_\Gamma=(1-T)\beta_\Gamma$.
\endproof

\smallskip

If we assume, as above, that the external momenta $p$ in the polynomial $\cP_\Gamma(x,p)$
are rational, then the form $\eta_\Gamma=\eta_{\Gamma,D,\ell,p}(x)$ is an algebraic differential
form defined over $\Q$, hence we can also assume that the 
form with logarithmic poles $\beta_\Gamma$ is also defined over $\Q$.

\smallskip

In addition to the integral \eqref{renormXint}, one also has the collection of the iterated
Poincar\'e residues along the  intersections of components of the divisor $\cZ_\ell$. Namely, 
for any $\cZ_{I,\ell}=\cap_{j\in I} Z_{j,\ell}$, with $Z_{j,\ell}$ the components of $\cZ_\ell$,
we have the additional integrals
\begin{equation}\label{ResZI}
\cR(\Gamma)_I = \int_{\tilde\Upsilon(\sigma_n)\cap \cZ_{I,\ell}\cap X}   {\rm Res}_{\cZ_I}(\beta_\Gamma).
\end{equation}

\medskip
\subsubsection{Griffiths-Schmid approach}\label{GSrenSec}

A global replacement of $\eta_\Gamma$ by a single form $\beta_{\Gamma,D,\ell,p}$ 
on $K\GL_\ell$ with logarithmic poles along $\cZ_\ell$ can be obtained if we use the
$\cC^\infty$-logarithmic de Rham complex instead of the algebraic or analytic one.

\begin{prop}\label{CinftyRen}
Consider the class $[\eta_\Gamma]$ in the analytic de Rham cohomology
$H^*_{dR}(\GL_\ell;\C)$. 
There is a $\cC^\infty$-form $\beta^{\infty}_\Gamma$ on $K\GL_\ell$ with logarithmic 
poles along $\cZ_\ell$ such that 
\begin{equation}\label{CinftybetaHstar}
 [\beta^\infty_\Gamma]=[\eta_\Gamma] \in H^*_{dR}(K\GL_\ell \smallsetminus \cZ_\ell;\C)=
H^*_{dR}(\GL_\ell;\C).
\end{equation}
Applying the Birkhoff factorization yields a renormalized integral
\begin{equation}\label{renormCintfy}
R^\infty(\Gamma) = \int_{\tilde\Upsilon(\sigma_n)} \beta^{\infty,+}_{\Gamma,D,\ell,p},
\end{equation}
where $\beta^{\infty,+}_\Gamma$ is a simple pole subtraction on $\beta^\infty_\Gamma$,
and iterated residues
\begin{equation}\label{ResZICinfty}
\cR^\infty(\Gamma)_I = \int_{\tilde\Upsilon(\sigma_n)\cap \cZ_{I,\ell}} {\rm Res}_{\cZ_I} (\beta^\infty_\Gamma).
\end{equation}
\end{prop}

\proof For $\cX$ a complex smooth projective
variety and $\cZ$ a simple normal crossings divisor, let 
$\Omega_{\cC^\infty(\cX)}(\log \cZ)$  be 
the $\cC^\infty$-logarithmic de Rham complex. The Griffiths-Schmid theorem
(Proposition 5.14 of \cite{GriSch}) shows that there is 
an isomorphism $H^*_{dR}(\cU) = H^*(\Omega_{\cC^\infty(\cX)}(\log \cZ))$.
\endproof

\smallskip

\begin{rem}\label{noalgrem}{\rm 
With the Griffiths-Schmid theorem one looses the algebraicity
of differential forms. Namely, the forms $\beta^\infty_\Gamma$ 
and $\beta^{\infty,+}_\Gamma$ are only smooth and not algebraic
or analytic differential forms. Even if the resulting form 
$\beta^{\infty,+}_\Gamma$, after pole subtraction, can then
be replaced by an algebraic de~Rham form in the same cohomology
class in $H^*_{dR}(K\GL_\ell;\C)$, it will remain, in general, only 
a form with $\C$-coefficients and not one defined over $\Q$. 
Thus, following this approach one
obtains a consistent renormalization procedure, but one can lose
control on the description of the resulting integrals as periods
of motives defined over a number field. }\end{rem}

\smallskip
\subsubsection{The Hodge filtration approach}\label{HodgerenSec}

There is another case in which a form can be replaced globally by
a cohomologous one with logarithmic poles on the complement of 
a normal crossings divisor, while only using algebraic
or analytic forms. Indeed, there is a particular piece of the
de Rham cohomology that is always realized by global sections of the
(algebraic) logarithmic de Rham complex. This is the piece
$F^n H^n_{dR}(\cU)$ of the Hodge filtration of Deligne's
mixed Hodge structure, with $n=\dim \cX$. This Hodge
filtration on $\cU$ is given by 
$$ F^pH^k_{dR}(\cU)={\rm Im}(\H^k(\cX,\Omega_\cX^{\geq p}(\log\cZ)) \to
\H^k(\cX, \Omega^\star_\cX (\log \cZ)) ). $$

\begin{prop}\label{FnHnLog}
Let $\cX$ be a smooth projective variety with $N=\dim \cX$, and let
$\cZ$ be a simple normal crossings divisor with affine complement
$\cU=\cX\smallsetminus \cZ$. Then, for $n\leq N$, the Hodge filtration satisfies
\begin{equation}\label{Hodge} 
H^0(\cX, \Omega^n_\cX(\log \cZ))= F^n H^n_{dR}(\cU) / F^{n+1} H^n_{dR}(\cU).
\end{equation}
In the case where $n=N$ the right-hand-side is reduced to $F^N H^N_{dR}(\cU)$.
\end{prop}

\proof   The Hodge filtration $F^p H^k_{dR}(\cU)$ is induced by the
naive filtration on $\Omega^\star_\cX (\log \cZ)$. Recall that (see
Theorem 8.21 and Proposition 8.25 of \cite{Voisin})
the spectral sequence of a filtration $F$ on a complex $K^\star$ that comes 
from a double complex $K^{p,q}$, with
$$ F^p K^n =\oplus_{r\geq p, r+s=n} K^{r,s} $$
has terms
$$ E^{p,q}_0 = {\rm Gr}^F_p K^{p+q} =F^p K^{p+q}/ F^{p+1} K^{p+q} = K^{p,q} $$
$$ E^{p,q}_1 = H^{p+q}(  {\rm Gr}^F_p K^\star ) = H^q( K^{p,\star} ) $$
$$ E^{p,q}_\infty = {\rm Gr}^F_p H^{p+q}( K^\star ) . $$
The spectral sequence above, applied to the 
Hodge filtration $F^p H^k_{dR}(\cU)$, is referred to as
the Fr\"olicher spectral sequence. It has 
$$ E^{p,q}_1 = H^q( \cX, \Omega_\cX^p (\log \cZ))   $$
$$ E^{p,q}_\infty = F^p H^{p+q}_{dR}(\cU) / F^{p+1} H^{p+q}_{dR}(\cU). $$
In particular, $E^{n,0}_1=H^0(\cX, \Omega_\cX^n (\log \cZ))$ and
$E^{n,0}_\infty= F^n H^n_{dR}(\cU) / F^{n+1} H^n_{dR}(\cU)$. When $n=N=\dim \cX$, 
the term $F^{N+1} H^N_{dR}(\cU)$ vanishes for dimensional reasons.

Deligne proved in \cite{Del} (see also the formulation of the
result given in Theorem 8.35 of \cite{Voisin}) that, in the case 
where $\cZ$ is a normal crossings divisor, the Fr\"olicher spectral sequence 
of the Hodge filtration degenerates at the $E_1$ term. Thus,
in particular, we obtain \eqref{Hodge}.
\endproof

\smallskip

\begin{cor}\label{inFn}
Given a meromorphic form $\eta$ with $[\eta]\in F^n H^n_{dR}(\GL_\ell)/F^{n+1} H^n_{dR}(\GL_\ell)$,  
with $n\leq \ell^2=\dim K\GL_\ell$, there is a form $\beta$ on $K\GL_\ell$ with logarithmic
poles along the normal crossings divisor $\cZ_\ell$, such that
\begin{equation}\label{betaHstar}
 [\beta]=[\eta] \in H^n_{dR}(K\GL_\ell \smallsetminus \cZ_\ell)=
H^n_{dR}(\GL_\ell).
\end{equation}
Then after pole subtraction one obtains 
\begin{equation}\label{UrenintDet}
 \int_{\tilde\Upsilon(\sigma_n)} \beta^+ ,
\end{equation}
which is a period of $\m(K\GL_\ell,\Sigma_{\ell,g})$.
\end{cor}

In this case also, in addition to the integral \eqref{UrenintDet}, we also
have the iterated residues (which in this case exist globally),
\begin{equation}\label{ResZIglob}
 \int_{\tilde\Upsilon(\sigma_n)\cap \cZ_{I,\ell}} {\rm Res}_{\cZ_I} (\beta).
\end{equation}

\smallskip

In general, it is difficult to estimate where the form $\eta_\Gamma$ lies in the Hodge filtration.
One can give an estimate, based on the relation between the filtration by order of pole and
the Hodge filtration, but it need not be accurate because exact forms can cancel higher order
poles. The same issue was discussed, in the original formulation in the graph hypersurface 
complement, in \S 9.2 and Proposition 9.8 of \cite{BlochKreimer}.

\smallskip

Let $\cX$ be a smooth projective variety and $\cZ \subset \cX$ a simple normal crossings divisor. 
As before, let $\cM_{\cZ,\cX}^\star$ denote the complex of meromorphic differential forms on $\cX$ 
with poles (of arbitrary order) along $\cZ$. This complex has a filtration $P^\star \cM_{\cZ,\cX}^\star$ 
by order of poles (polar filtration), where $P^k \cM^m_{\cZ,\cX}$ consists of the $m$-forms with pole 
of order at most $m-k+1$, if $m-k \geq 0$ and zero otherwise. Deligne showed in \S II.3, 
Proposition 3.13 of \cite{Del2} and Proposition 3.1.11 of \cite{Del}, that the filtration induced on 
the subcomplex $\Omega_\cX^\star(\log \cZ)$ by the polar filtration on $\cM_{\cZ,\cX}^\star$ is the
naive filtration (that is, the Hodge filtration), and that the natural morphism
$$ (\Omega_\cX^\star(\log \cZ), F^\star) \to (\cM_{\cZ,\cX}^\star, P^\star) $$
is a filtered quasi-isomorphism.  In particular (Theorem 2 of \cite{DeDi}) the image of $\H^\star(\cX, P^k \cM^\star_{\cX,\cZ})$
inside $H^\star_{dR}(\cU)$ contains $F^k H^\star_{dR}(\cU)$. This means that
we can use the order of pole to obtain at least an estimate of the position of $[\eta_\Gamma]$ in the Hodge filtration. 
We need to compute the order of pole of the pullback of the form $\eta_\Gamma$ along the blowups in the
construction of the compactification $K\GL_\ell$.

\smallskip

\begin{prop}\label{orderpoleE}
For a graph $\Gamma$ with $n=\# E_\Gamma$ and $\ell= b_1(\Gamma)$, such that
$n\geq \ell-2$, and with spacetime dimension $D\in \N$,
the position of $[\eta_\Gamma]$ in the Hodge filtration $F^k
H^n_{dR}(\GL_\ell)$
is estimated by $k \geq n-(\ell-1) (-n +(\ell+1)D/2) +(\ell-1)^2$.
\end{prop}

\proof At the first step in the construction of the compactification $K\GL_\ell$ we blow up the locus
of matrices of rank one. We need to compare the order of vanishing of $\det(x)^{-n + (\ell+1)D/2}$
along this locus, with the order of zero acquired by the form $\omega_\Gamma$
along the exceptional divisor of this blowup. The determinant vanishes at
order $\ell-1$ on that stratum. The form $\omega_\Gamma$, on the other hand, 
acquires a zero of order $c-1$ where $c$ is the codimension of the blowup locus. 
This can be seen in a local model: when blowing up
a locus $L=\{ z_1=\cdots =z_c=0 \}$ in $\C^N$, the local coordinates $w_i$
in the blowup can be taken as $w_i w_c=z_i$ for $i<c$ and $w_i=z_i$ for $i\geq c$, with
$E=\{ w_c=0 \}$ the exceptional divisor. Then for $n\geq c$, and a form $dz_1 \wedge
\cdots \wedge dz_n$, the pullback satisfies
$$ \pi^*(dz_1 \wedge \cdots \wedge dz_n) = d(w_c w_1)\wedge \cdots \wedge
d(w_c w_{c-1})
\wedge d(w_c) \wedge \cdots \wedge d(w_n) = w_c^{c-1} dw_1 \wedge \cdots
\wedge dw_n . $$
The codimension of the locus of rank one matrices is $c=(\ell-1)^2$.
Thus, when performing the first blowup in the construction of $K\GL_\ell$,
the pullback of the form $\eta_\Gamma$ acquires a pole of order 
$(\ell-1) (-n +(\ell+1)D/2) -(\ell-1)^2+1$
along the exceptional divisor. Further blowups do not alter this pole
order, hence we can estimate that the pullback of the $n$-form 
$\eta_\Gamma$ to the Kausz compactification is in the term $P^k$
of the polar filtration, with $n-k+1=(\ell-1) (-n +(\ell+1)D/2)
-(\ell-1)^2+1$. Taking into account the possibility of reductions of the order 
of pole, due to cancellations
coming from exact forms, we obtain an estimate for the position in the
polar and in the
Hodge filtration, with $k \geq n-(\ell-1) (-n +(\ell+1)D/2) +(\ell-1)^2$.
\endproof

\smallskip
\subsection{Nature of the period}

We then discuss the nature of the period obtained by the evaluation of
\eqref{UrenintDet}. We need a preliminary result. We define a mixed Tate
configuration $Y$ in an ambient variety $X$ as follows. 

\begin{defn}\label{MTconfigDef}
Let $X$ be a smooth projective variety and $Y\subset X$ a divisor
with irreducible components $\{ Y_i \}_{i=1}^N$. Let $\cC_Y=\{ Y_I=Y_{i_1}\cap \cdots \cap Y_{i_k}\,|\,
I=(i_1,\ldots,i_k),\, k\leq N \}$.  Then $Y$ is
a mixed Tate configuration if all unions $Y_{I_1}\cup \cdots\cup Y_{I_r}$
of elements of the set $\cC_Y$ have motives $\m(Y_{I_1}\cup \cdots\cup Y_{I_r})$
contained in the Voevodsky derived category of mixed Tate motives.
\end{defn}

\begin{rem}\label{MTconfrem}{\rm
Note that in Definition \ref{MTconfigDef} we do not require that $Y$ is 
necessarily a normal crossings divisor.  However, in the case of the boundary divisor 
$\cZ_\ell$ of $K\GL_\ell$, we will use in Proposition \ref{ramiprop} the fact that it is 
also normal crossings, in addition to satisfying the condition of Definition~\ref{MTconfigDef}
(see Lemma~\ref{MTconfigZ}).  
}\end{rem}

Let $\Sigma_{\ell,g}$ be the proper transform of the divisor given by
the projective version of $\hat\Sigma_{\ell,g}$
described in \eqref{Sigmaellg}, defined by the same equations.

\begin{lem}\label{DivisorTate}
The divisor $\Sigma_{\ell, g}$ is a mixed Tate configuration.
\end{lem}

\proof By \eqref{Sigmaellg}, $\Sigma_{\ell, g}$ and any arbitrary
union of components are hyperplane arrangements. 
It is known from \cite{BGSV} that motives of hyperplane arrangements
are mixed Tate, see also \S 1.7.1--1.7.2 and \S 3.1.1 of \cite{Dup}, where  
the computation of the motive in the Voevodsky 
category can be obtained in terms
of Orlik--Solomon models.
Using a characterization of the mixed Tate condition in terms
of eigenvalues of Frobenius, the mixed Tate nature of hyperplane 
arrangements was also proved in Proposition 3.1.1 of \cite{KisLeh}.
The mixed Tate property can be seen very explicitly at the level 
of the virtual motive. In fact, the Grothendieck
class of an arrangement $A$ in $\P^n$ is explicitly given (Theorem
1.1. of  \cite{Alu}) by
$$ [A]=[\P^n] - \frac{\chi_{\hat A}(\bL)}{\bL-1}, $$
where $\chi_{\hat A}(t)$ is the characteristic polynomial of the
associated central arrangement $\hat A$ in $\A^{n+1}$. It then
follows by inclusion-exclusion in the Grothendieck ring that all
unions and intersections of components of $A$ are mixed Tate.
\endproof

\smallskip

\begin{lem}\label{MTconfigZ}
The boundary divisor $\cZ_\ell$ of the Kausz compactification $K\GL_\ell$
is a mixed Tate configuration. 
\end{lem}

\proof The motives of unions of intersections of components of $\cZ_\ell$ can be described
in terms of motives of bundles over products of flag varieties with fibers that are lower
dimensional compactifications $K\GL_i$ and $\overline{\PGL}_i$, and we proved in 
Proposition~\ref{Chowprop} (see also Propositions~\ref{KauszTateGr} and \ref{nummotive})
that all these motives are Tate, hence the condition of Definition~\ref{MTconfigDef} is
satisfied.
\endproof

\smallskip

\begin{prop}\label{MTMperiod}
When the form $\beta_{\Gamma, D,\ell,p}$ on the big cell extends to a logarithmic form in 
$\Omega_{K\GL_\ell}^\star(\log \cZ_\ell)$, the integral 
$R(\Gamma) = \int_{\tilde\Upsilon(\sigma_n)} \beta^+_{\Gamma, D,\ell,p}$ 
is a period of a mixed Tate motive.
\end{prop}

\proof   In the globally defined case, 
this is an integral of an algebraic differential form defined
on the compactification $K\GL_\ell$, hence
a genuine period, in the sense of algebraic geometry, of $K\GL_\ell$.
By Proposition \ref{Chowprop}, we know that the Chow motive
$h(K\GL_\ell)$ is Tate. We also know from Lemma \ref{DivisorTate} 
that the motive $\m(\Sigma_{\ell,g})$ is mixed Tate. 
Under the embedding of pure motives into mixed motives we
obtain objects $\m(K\GL_\ell)$ and $\m(\Sigma_{\ell,g})$ in the derived
category of mixed Tate motives, 
$\cD\cM\cT\cM(\Q)$, that is, the smallest triangulated subcategory of the
Voevodsky triangulated category of mixed motives $\cD\cM(\Q)$ containing
all the Tate objects $\Q(n)$. It then follows that the relative
motive $\m(K\GL_\ell, \Sigma_{\ell,g})$ is also mixed Tate, as it sits in a distinguished
triangle in the Voevodsky triangulated category, where the other two terms are
mixed Tate. 
\endproof

\smallskip

\begin{rem}\label{heart}{\rm 
In the proof of Proposition \ref{MTMperiod} here above, we viewed the motive 
$\m(K\GL_\ell, \Sigma_{\ell,g})$ as an element in the derived category 
$\cD\cM\cT\cM(\Q)$ of mixed Tate motives.  All the varieties we are considering 
are defined over a number field, in fact over $\Q$. In the number field case, 
an abelian category of mixed Tate motives can be constructed as the heart 
of a $t$-structure in $\cD\cM\cT\cM(\Q)$: this is possible because the Beilinson-Soul\'e vanishing 
conjecture holds over a number field, see \cite{Levine}.  We denote by $\cM\cT\cM(\Q)$ 
this abelian category of mixed Tate motives. To obtain objects in $\cM\cT\cM(\Q)$ one
applies the cohomology functor with respect to the $t$-structure. For example, for a
projective space $\P^n$, one has the motive $\m(\P^n)=\Q(0) \oplus \Q(-1)[2]
\oplus \cdots \oplus \Q(-n)[2n]$ in $\cD\cM\cT\cM(\Q)$; its cohomology with
respect to the $t$-structure is $h_{2i}(\m(\P^n))=\Q(-i)$, which is an object in
$\cM\cT\cM(\Q)$, while the shifts are not. In the following, with an abuse
of notation, we will write the motive simply as $\m(K\GL_\ell, \Sigma_{\ell,g})$, and more
generally $\m(K\GL_\ell \smallsetminus A, B)$ for the cases considered in
Proposition~\ref{ramiprop}, although when we refer to the motive in $\cM\cT\cM(\Q)$
what we are really considering are the pieces of the cohomology with
respect to the $t$-structure, and in particular, for the conclusion about the period,
the piece that corresponds to the degree of the differential form.
}\end{rem}

\smallskip

\begin{prop}\label{MTMlocal}
Let $\beta_{\Gamma, D,\ell,p}$ be the form with logarithmic poles on the top
cell $X$ of the cellular decomposition of $K\GL_\ell$, as in Corollary \ref{intcellslog}.
If the motive $\m(\Sigma_{\ell,g}\cap X)$ is mixed Tate, then
the integral $R(\Gamma) = \int_{\tilde\Upsilon(\sigma_n)} \beta^+_{\Gamma, D,\ell,p}$ 
is a period of a mixed Tate motive.
\end{prop}

\proof Using distinguished triangles in the Voevodsky category, we see that, 
if the motive $\m(\Sigma_{\ell, g}\cap X)$ is mixed Tate,
then the motive $\m(X, \Sigma_{\ell, g}\cap X)$ also is, since the big cell 
has $\m(X)=\bL^{\ell^2}$. The result then follows, since the integral is
by construction a period of the motive $\m(X, \Sigma_{\ell, g}\cap X)$.
\endproof

\smallskip

\begin{rem}\label{newdiv}{\rm
The central difficulty in the approach of \cite{AluMa}, which was to analyze the
nature of the motive of $\m(\Sigma_{\ell, g}\cap \cD_\ell)$,
is here replaced by the problem of identifying the nature of the motive $\m(\Sigma_{\ell, g}\cap X)$,
where $X$ is the big cell of $K\GL_\ell$. }
\end{rem}

It may seem at first that we have simply substituted the problem of understanding for which
range of $(\ell, g)$ the intersection of the divisor $\Sigma_{\ell, g}$ with $\GL_\ell$ remains 
mixed Tate, with the very similar problem of when the intersection of $\Sigma_{\ell, g}$ with
the big cell $X$ of $K\GL_\ell$ remains mixed Tate. However, this reformulation makes it
possible to use the explicit description of the cells $X(\lambda,x)$ of spherical varieties in 
terms of limits as in \eqref{Xlambdax}, to analyze this question. We do not discuss this
further in the present paper, and we simply state it as a question.

\begin{ques}\label{BigcellSigmaques}
Let $X$ be the big cell of the cellular decomposition \eqref{Xlambdax} of $K\GL_\ell$
and let $\Sigma_{\ell,g}$ be the divisor described above. For which pairs $(\ell,g)$ is the motive
$\m(\Sigma_{\ell,g}\cap X)$ mixed Tate? 
\end{ques}

\smallskip

One defines the category $\cM\cT\cM(\Z)$ of mixed Tate motives over $\Z$ as
mixed Tate motives in $\cM\cT\cM(\Q)$ that are unramified over $\Z$.
An object of $\cM\cT\cM(\Q)$ is unramified over $\Z$ if and only if, for any
prime $p$, there exists a prime $\ell\neq p$ such that the $\ell$-adic realization 
is unramified at $p$, see Proposition 1.8 of \cite{DG}. In the following statement 
our notation for the motives is meant as in Remark~\ref{heart} above. 

\smallskip

\begin{prop}\label{ramiprop}
The motives $\m(K\GL_\ell)$ are unramified over $\Z$. More generally,
if $A$ and $B$ are unions of two disjoint sets of components of the
boundary divisor $\cZ_\ell$ of the compactification $K\GL_\ell$, the
motives $\m(K\GL_\ell \smallsetminus A, B)$ are unramified over $\Z$.
\end{prop}

\proof This question can be approached in a way analogous to our
previous discussion of the Chow motive, namely using the description
of $K\GL_\ell$ as an iterated blowup and the properties of the divisor
of the compactification. 
The argument is similar to the one used in Theorem~4.1 and Proposition~4.3 of 
\cite{GonMan} to prove the analogous statement for the moduli spaces 
$\overline{\cM}_{0,n}$ of rational curved with marked points. 
There, it is shown that $\overline{\cM}_{0,n}$ is unramified over $\Z$ 
by showing that the combinatorics of the  normal crossings divisor of the 
compactification is not altered by reductions mod $p$, see Definition~4.2 
of \cite{GonMan}. For a pair $(\cX,\cZ)$, with $\cZ\subset \cX$ a normal crossing divisor, 
the condition that the reduction mod $p$ does not
alter the combinatorics means that $\cX$ and all the strata of
$\cZ$ are smooth over $\Z_p$ and the reduction mod $p$ gives 
a bijection from the strata of $\cZ$ to those of the special fiber (see Definition 4.2
of \cite{GonMan} and Definition 3.9 of \cite{Gon2}). 
In our case we have $\cX=K\GL_\ell$, with $\cZ=\cZ_\ell$ the 
boundary divisor of the Kausz compactification. As we showed in \S\S \ref{virtKsec},
\ref{numKsec}, and \ref{chowKsec}, the motive $\m(K\GL_\ell)$ is a Tate motive. More
generally, the motives $\m(K\GL_\ell \smallsetminus A, B)$ are mixed Tate over $\Q$: 
this can be seen as in Proposition 3.6 of \cite{Gon2}, using Lemma~\ref{MTconfigZ},
which shows that $\cZ_\ell$ is both a normal
crossings divisor and a mixed Tate configuration in the sense of Definition \ref{MTconfigDef}.
This implies, by the argument of Proposition 3.6 of
\cite{Gon2}, that all the motives $\m(K\GL_\ell \smallsetminus A, B)$ are mixed Tate over $\Q$.
To check the condition that the reduction map preserves the combinatorics of $(K\GL_\ell ,\cZ_\ell)$,
first note that both $K\GL_\ell$ and the strata of the normal crossing divisor $\cZ_\ell$ are 
smooth over $\Z$, by Corollary 4.2 \cite{Kausz}. Moreover, the 
description of the Kausz compactification and of the strata of its boundary 
divisor given in Theorems 9.1 and 9.3 of \cite{Kausz} also holds 
over fields of characteristic $p$ and is compatible with reduction, so that
the set of strata is matched under the reduction map. 
The argument of Proposition~4.3 of \cite{GonMan} showing that the $\ell$-adic 
realization is then unramified, for all $\ell$ with $\ell\neq p$, is based on the 
argument of Proposition 3.10 of \cite{Gon2}. Following this reasoning, the cohomologies 
$H^*(\cX \smallsetminus A, B)$ can be computed using a simplicial resolution 
$\cS_\bullet(\cX \smallsetminus A, B)$, whose simplexes correspond to unions of
intersections of components of the divisor. The argument of Proposition 3.10 of \cite{Gon2} then
shows that the reduction map applied to the simplicial schemes $\cS_\bullet(\cX \smallsetminus A, B)$   
induces an isomorphism in \'etale cohomology, $H^*_{et}(\bar\cX\smallsetminus \bar A, \bar B, \Q_\ell)
\simeq H^*_{et}(\bar\cX^0\smallsetminus \bar A^0, \bar B^0, \Q_\ell)$, where 
$\bar \cX=\cX \otimes_{\Z_p}\Q_p$ and $\cX^0$ is the special fiber of the reduction. This
shows that the \'etale realization is unramified for $\ell\neq p$. By Proposition 1.8 of \cite{DG}
this means that the motives $\m(K\GL_\ell \smallsetminus A, B)$ are mixed Tate over $\Z$.
\endproof

\begin{rem} {\rm
Given that the unramified condition holds, one can conclude from Brown's theorem \cite{Brown}
and the previous Proposition~\ref{MTMperiod} (and Proposition~\ref{MTMlocal}, when 
$\m(\Sigma_{\ell, g}\cap X)$ is mixed Tate) that the integral
\eqref{UrenintDet} is a $\Q[\frac{1}{2\pi i}]$-linear combination
of multiple zeta values.}
\end{rem}

\smallskip
\subsection{Comparison with Feynman integrals}\label{FeynmanSec}

The result obtained in this way clearly differs from the usual
computation of Feynman integrals, where the 
methods used are based on regularization and pole subtraction
of the integral (dimensional regularization, cutoff, zeta regularization, etc.)
There are several reasons behind this difference, which we now discuss briefly.

\smallskip

In the usual physical renormalization non-mixed-Tate periods
are known to occur, \cite{BrDo}, \cite{BrSch}. In the setting we discussed
here, the only possible source of non-mixed-Tate cases is the motive of the 
intersection $\Sigma_{\ell,g} \cap X$, where $X$ is the big cell of the Kausz
compactification $K\GL_\ell$. In particular, this locus is the
same for all graphs with fixed loop number $\ell$ and fixed genus $g$. 
However, in the usual physical renormalization, not all graphs with the 
same $\ell$ and $g$ have periods of the same nature, as one can see
from the examples analyzed in \cite{Dor}, \cite{Schn}.

\smallskip

There is loss of information in mapping the computation of the
Feynman integral from the complement of the graph
hypersurface (as in \cite{BEK}, \cite{BrDo}, \cite{BrSch})
to the complement of the determinant hypersurface (as in \cite{AluMa}),
when the combinatorial conditions on the graph recalled in \S \ref{DetFeynSec}
are not satisfied. Explicit examples of graphs that violate those
conditions are given in \S 3 of \cite{AluMa}. In such cases the map \eqref{Upsilonmap} 
need not be an embedding, hence part of the information contained
in the Feynman integral calculation \eqref{paramInt} will be lost
in passing to \eqref{UintDet}.

\smallskip

However, this type of loss of information does not affect some of
the cases where non-mixed Tate motives are known to appear in
the momentum space Feynman amplitude.

\begin{ex}\label{Dorynex}{\rm
Let $\Gamma$ be the graph with $14$ edges that gives a counterexample
to the Kontsevich polynomial countability conjecture, in Section 1 of \cite{Dor}.
The map $\Upsilon: \A^n \to \A^{\ell^2}$ of \eqref{Upsilonmap} 
has $n=\# E(\Gamma)=14$ and $\ell=b_1(\Gamma)=7$. Let $\Upsilon_i$
denote the composition of the map $\Upsilon$ with the projection onto the $i$-th row
of the matrix $M_\Gamma$ of \eqref{MGamma}. 
In order to check if the embedding condition for $\Upsilon$ is satisfied, we know from 
Lemma 3.1 of \cite{AluMa} that it suffices to check that $\Upsilon_i$ is injective
for $i$ ranging over a set of loops such that every edge of $\Gamma$ is part of a loop
in that set. This can then be checked by computer verification for the matrix 
$M_\Gamma$ of this particular graph. }
\end{ex}

The example above is a log divergent graph in dimension four. It is known
to give a non-mixed Tate contribution with the usual method of computation
of the Feynman integral, \cite{Dor}, \cite{BrDo}. 
The same verification method we used for this case can be applied to the other
currently known explicit counterexamples in \cite{Dor}, \cite{BrDo}, 
\cite{Schn}, \cite{BrSch}.

\smallskip

Even for integrals (including the example above) where the map \eqref{Upsilonmap}
is an embedding, the regularization
and renormalization procedure described here, using the Kausz
compactification and subtraction of residues for forms with
logarithmic poles, is not equivalent to the usual
renormalization procedures of the regularized integrals. For instance,
our regularized form (hence our regularized integral) can be trivial
in cases where the usual regularization and renormalization
would give a non-trivial result. This may occur if the form $\beta$ with
logarithmic poles happens to have a nontrivial residue, but a trivial 
holomorphic part $\beta^+$.

\smallskip

In such cases, part of the information loss coming from pole subtraction on
the differential form is compensated by keeping track of the
residues. However, in our setting these also deliver
only mixed Tate periods, so that even when this information
is included, one still loses the richer structure of the periods
arising from other methods of regularization and renormalization,
adopted in the physics literature.

\bigskip

\subsection*{Acknowledgment} The authors are very grateful to
the anonymous referee for many very detailed and helpful comments
and suggestions that greatly improved the paper.  
The first author was partially supported by
NSF grants DMS-1007207, DMS-1201512, and PHY-1205440.

\end{document}